\documentclass{article}

\usepackage[backend=biber,style=alphabetic,doi=true,isbn=false,url=false,maxbibnames=5]{biblatex}
\usepackage{mathtools}
\usepackage{changepage}
\usepackage{prelude}

\newcommand\nbd\nobreakdash
\long\def\onlyshort#1{}

\DeclareMathOperator{\des}{des}
\newcommand{\Lie}{\mathrm{Lie}}
\newcommand{\ellvec}{\boldsymbol{\ell}}
\newcommand{\e}{\mathrm{e}}
\newcommand{\LP}{\mathbf{LP}}

\newcommand{\mc}[1]{\mathcal{#1}}
\newcommand{\Htarg}{H_{\mathrm{target}}}
\newcommand{\hatHtarg}{\widehat{H}_{\mathrm{target}}}

\definecolor{myred}{RGB}{221,0,0}
\definecolor{myblue}{RGB}{0,0,221}
\definecolor{citegreen}{RGB}{0,165,0}
\definecolor{mygreen}{RGB}{18,143,31}
\definecolor{citeblue}{RGB}{35,55,225}

\begin{document}

\title{Floquet-Universal Hamiltonian Simulation}

\author[1,2]{Emilio Onorati}
\author[3]{Harriet Apel}
\author[1]{Michael M.~Wolf}
\author[3,*]{Toby Cubitt}
\affil[1]{Department of Mathematics, Technische Universit\"at M\"unchen, Germany}
\affil[2]{Department of Physics, Freie Universit\"at Berlin, Germany}
\affil[3]{Department of Computer Science, University College London, United Kingdom}

\date{}

\maketitle

\begingroup
\renewcommand{\thefootnote}{\fnsymbol{footnote}}
\footnotetext[1]{\mbox{email: \href{mailto:t.cubitt@ucl.ac.uk}{\texttt{t.cubitt@ucl.ac.uk}}}}
\endgroup

\begin{abstract}
  
Analogue and digital Hamiltonian simulation represent two distinct approaches to making one quantum system replicate the physical properties of another, with consequences ranging from practical applications to complexity theory and quantum gravity.
They are rooted in two distinct regimes of quantum dynamics: time-independent Hamiltonians for the former; fully controllable time-dependent quantum circuits in the latter.
Time-periodic Hamiltonians represent an intermediate regime.
In this work, we establish a theory of Floquet simulation where we use periodically driven Hamiltonians to synthesise time-independent ones.

We show that a set of interactions $\mathcal{S}$ can Floquet-simulate every Hamiltonian in the Lie algebra $\mathrm{Lie}(\mathcal{S})$, and consequently we give a complete and constructive characterisation of Floquet-universal Hamiltonians that are able to produce any target Hamiltonian.
Notably, our construction uses only $\mathcal{O}(1)$ local interaction strengths and ratios thereof, avoiding the impractical multi-scale local interaction strengths often required in analogue simulation via time-independent Hamiltonians.
Instead, the construction requires a range of driving frequencies which scales polynomially with system size for important classes such as $k$-local lattice Hamiltonians.
This has immediate complexity-theoretic implications, including numerous BQP-completeness and QMA-hardness results for natural physical properties and quantities studied in Floquet physics.

\end{abstract}

 \clearpage

\tableofcontents
\clearpage

\part{Overview}
\section{Introduction}
Quantum simulation is both of immense practical value and fundamental scientific significance.
On the practical front, simulating the quantum many-body systems that occur in chemistry and materials science is one of the most promising applications of quantum technology~\cite{Daley2022}.
Analogue simulation, in which a quantum system is engineered to reproduce a target Hamiltonian of interest, has been successfully performed for decades, at scales well-beyond what is possible on classical or quantum computers~\cite{browaeys2020many,monroe2021programmable}.
Meanwhile, recent results have reported the first scientifically meaningful digital quantum simulations that challenge current state-of-the-art classical techniques~\cite{AlamGoogle2025,AlamQuantinuum2025}.
Fundamentally, quantum simulation has deep connections to quantum computation and complexity theory.
It was established long ago that general-purpose digital quantum simulation is theoretically equivalent to a universal quantum computer~\cite{Lloyd1996,Dodd2002}.
Meanwhile, analogue quantum simulation is closely related to Hamiltonian complexity theory~\cite{Cubitt2019,CubittMontanaro2016,BravyiHastings2017}, with recent theoretical results proving an exact equivalence between general-purpose analogue quantum simulators and Hamiltonians in the QMA complexity class~\cite{Kohler2022}.

These qualitatively different sets of simulation experiments and theorems, analogue and digital, represent two distinct mathematical structures underlying the simulation.
In analogue simulation, a static, time-independent Hamiltonian is engineered to replicate the desired target Hamiltonian.
This form of simulation was originally proposed in~\cite{Jaksch1998}, and a full mathematical theory underlying it worked out in~\cite{Cubitt2019}.
Meanwhile, in digital quantum simulation, a quantum circuit consisting of a sequence of few-qubit unitary gates is constructed, which reproduces the time-dynamics of the target Hamiltonian.
Its theory is exactly that of universal quantum computation and the Hamiltonian simulation algorithm, worked out in the 1980-90s~\cite{NielsenChuang2010}.

Thus analogue quantum simulation (at least in its original form) uses a time-independent Hamiltonian with no dynamically varying parameters to carry out the simulation.
Whereas digital quantum simulation uses a quantum computer, which can be viewed as a highly time-dependent Hamiltonian with many dynamically varying parameters.

The consequences of this underlying mathematical difference can be traced all the way through from theory to state-of-the-art experiments.
On the theoretical side, it is a classic result from the early days of quantum computing theory that any entangling two-qubit interaction (together with single-qubit gates) is sufficient to carry out universal quantum computation~\cite{Dodd2002}, hence sufficient for digital simulation of any desired target Hamiltonian.
Whereas the classification of two-qubit interactions sufficient for universal analogue simulation is non-trivial, does not encompass all entangling interactions, and was worked out only in 2018~\cite{Cubitt2019}.
Experimentally, because analogue simulation only requires engineering time-independent Hamiltonians, it is less demanding of hardware, hence far larger analogue quantum simulation experiments have been performed than digital.
On the other hand, since only time-independent Hamiltonians are available, the range of target Hamiltonians that can feasibly be simulated using analogue techniques is more restricted.

There is an interesting intermediate regime between time-independent Hamiltonians with no dynamical control, and time-dependent Hamiltonians with full dynamical control: periodically driven Hamiltonians, also called ``Floquet Hamiltonians''.
Their time-dependence is periodic:
\begin{equation}
  H(t) = H(t+T)
\end{equation}
for some period $T$ and therefore can be expressed as a Fourier series:
\begin{equation}
  H(t) = \sum_{k=1}^\infty a_k \cos(\omega k t) + b_k \sin(\omega k t).
\end{equation}
The Floquet Hamiltonians encountered in practice typically have simple Fourier series, consisting of a finite number of terms not involving extremely large values of $k$; i.e.\ Fourier series with bounded bandwidth.

As with analogue quantum simulation, Floquet quantum systems have been engineered and studied experimentally for multiple decades.
They exhibit interesting many-body physics phenomena, many of which parallel those of static many-body quantum systems~\cite{Bukov2015,Khemani2016}, but some of which are not possible with time-independent Hamiltonians, most famously time-crystals~\cite{Else2016,GoogleTimeCrystal}.
Simulation of Floquet systems using digital quantum computers has become popular recently~\cite{GoogleTimeCrystal,Zhang2022FloquetSPT}, largely because the circuit depths attainable in current quantum computers are typically insufficient to achieve a good approximation to continuous-time dynamics.
Although these are discrete Floquet systems consisting of repeated layers of quantum gates, rather than the sinusoidally periodically-driven Hamiltonians of conventional Floquet physics.

However, the use of Floquet Hamiltonians themselves to simulate other quantum Hamiltonians remains less explored  than either the static or fully time-dependent cases.
In this paper, we initiate the systematic, mathematical study of quantum Floquet simulation: specifically, the use of periodically driven Hamiltonians to simulate target Hamiltonians that are time-independent.
We prove a complete characterisation of the Floquet-simulation power of any set of interactions, and as a consequence fully characterise Floquet-universal Hamiltonians able to simulate any time-independent Hamiltonians.

The classification is simple to state, although its proof is non-trivial: a set of periodically driven interactions $\cS$ can Floquet-simulate any Hamiltonian in the Lie algebra $\Lie(\cS)$ generated by $\cS$.
In particular, this means that a set of interactions can simulate all target Hamiltonians -- is ``Floquet universal'' -- if it generates the full Lie algebra.

Comparing this characterisation with the time-independent and fully time-dependent cases, we see that the condition for Floquet-universality coincides with the classification of interactions that generate universal gates for quantum computation.
Hence, restricting from discontinuous time-dependent Hamiltonians with infinite frequency bandwidth to smoothly time-periodic Hamiltonians with bounded bandwidth does not reduce the computational or simulation power.
This contrasts to restricting to time-independent control, which was shown in~\cite{Cubitt2019,Kohler2022} to in general reduce both.

Despite arriving at the same classification, the proof in the Floquet case is very different and significantly more involved than the case of universal quantum gates.
It uses a novel algebraic geometry argument via residues to relate the effective Hamiltonian to the canonical projection from the free associative algebra to the Lie algebra, and requires a recent tighter bound on the tail of the Magnus expansion~\cite{ACO}.
Our characterisation of the Floquet-simulation power of an arbitrary set of interactions is also constructive: given any target Hamiltonian, suitable periodic driving parameters for generating that Hamiltonian can be derived by essentially a linear algebra computation.

Notably, all these results still hold when all ratios between local interaction strengths in the engineered Floquet Hamiltonian bounded between 1 and~2 (in natural units) so that only the driving frequencies are allowed to scale.
This is a key practical motivation for this work, contrasting strongly with the time-independent case where the the ratios of the interaction strengths scale polynomially, and there are no-go results indicating that this cannot be improved~\cite{Cubitt2019,AharonovZhou2019,Harley2024,Baspin2025}.
For this reason, analogue Hamiltonian simulation constructions based on perturbation gadgets or history state constructions have had little practical impact, as they require precise engineering of local interaction strengths across different energy scales, well beyond capabilities of analogue experiments.
On the other hand, engineering finite bandwidth periodic driving is potentially more experimental feasible -- indeed, this type of driving has already been realised experimentally~\cite{Jotzu2014,Choi2020,Koyluoglu2025}.

Since universality encompasses the simulation of arbitrary Hamiltonians, including families with exponentially many independent parameters and QMA hard instances, unsurprisingly the driving frequencies in the general Floquet-universal construction can grow exponentially with the size of the target Hamiltonian.
The relevant question is whether an efficient construction exists for Hamiltonians with a succinct description.
We show that it does for the important class of $k$-local Hamiltonians on a lattice, and more generally $k$-local Hamiltonians whose non-commutativity graph has constant chromatic number.
For any such Hamiltonian on $n$ qudits, Floquet simulation can be achieved with all interaction strengths $\mc O(1)$ and all driving frequencies bounded by $\operatorname{poly}(n)$.

In the next section we will review Floquet theory before defining Floquet simulation and establishing its broader consequences.
We then present our two main results in \cref{sect:main results}: the characterisation of the simulation power of Floquet Hamiltonians including a description of Floquet-universal Hamiltonians; and $k$-local lattice Hamiltonians can be efficiently simulated.
The complete proof and rigorous statement of our results is deferred to \cref{part:technical}.

\section{Floquet theory}

We will be concerned throughout with time-dependent Hamiltonians that are periodic in time, which we capture in a definition for concreteness:

\begin{definition}
  A time-dependent Hamiltonian is \emph{time-periodic} (or just \emph{periodic}) if there exists a $T$ such that $H(t)=H(t+nT)$ for all $n\in\N$.
  In this case, $T$ is the \emph{period} and $\omega = 2\pi/T$ is the \emph{frequency} of $H(t)$.
\end{definition}

Floquet engineering is an application of the Floquet-Bloch theorem \cite{Floquet1883,Bloch1929} which describes how a periodically driven quantum system can be transformed into an effective time-independent dynamics by a time-dependent unitary change of basis.

\begin{theorem}[Floquet Theorem: see \cite{Floquet1883,Bloch1929}]\label{thm:Floquet-Bloch}
  Let $H(t)$ be a time-periodic Hamiltonian, and $U(t)$ its associated time-evolution operator.
  For all $t\in[0,T]$
  \begin{equation}
    U(t +nT) =U(t)\left[U(T) \right]^n .
  \end{equation}
\end{theorem}
It is therefore useful to define a time-independent Floquet Hamiltonian whose time-evolution operator coincides with that of the periodic Hamiltonian at integer multiples of the period.
\begin{definition}[Floquet Hamiltonian]\label{def:Floquet_Hamiltonian}
  Given a periodic Hamiltonian $H(t)$ with period $T$, the associated \emph{Floquet Hamiltonian}, $H_F$, satisfies
  \begin{equation}
    e^{-i H_F n T} = \exp\left[-i \mathcal{T} \int_{0}^{nT} H(s)ds \right]
  \end{equation}
  for any $n\in\N^+$.
  We write $H_F(n) = n H_F$ for the Floquet Hamiltonian defined over $n$ periods.
\end{definition}
Note the Floquet Hamiltonian is not unique; for example, integer multiples of $2\pi/T$ can be added to $H_F$ without altering the above equality.
There is also an additional gauge freedom in choosing when during the period the system starts; i.e.\ when considering a period $[t_0,t_0+T]$ we have freedom to chose $t_0\in[0,T]$~\cite{Bukov2015, Goldman, Saar2003}.
For simplicity, we disregard this Floquet gauge freedom throughout, equivalent to implicitly choosing $t_0=0$.

\section{Floquet simulation}

We first map out the rigorous mathematical formulation of simulation when the simulating system is periodic.
\begin{definition}[Floquet simulation]\label{def:Floquet_simulation}
  A time-periodic Hamiltonian $H(t)$ \emph{Floquet-simulates} the time-independent Hamiltonian $\Htarg$ to precision $\epsilon$ over $n$ periods -- or, more briefly, \emph{$(\epsilon,n)$-Floquet-simulates} $\Htarg$ -- if
  \begin{equation}
    \norm{H_F(n) - \Htarg} \leq \epsilon,
  \end{equation}
  where $H_F$ is its associated Floquet Hamiltonian.
  We say a Floquet simulation is \emph{perfect} if $\epsilon=0$.
\end{definition}

A simple and immediate consequence of the Floquet theorem is:
\begin{corollary}
A time periodic Hamiltonian ${H(t) = H(t+nT)}$ perfectly Floquet-simulates its associated Floquet Hamiltonian.
\end{corollary}
Similarly, a time-periodic Hamiltonian will approximately Floquet-simulate any time-independent Hamiltonian $H'$ that is close to its associated Floquet Hamiltonian:
\begin{corollary}
  Let $H(t)$ be a time-periodic Hamiltonian, and $H_F$ its associated Floquet Hamiltonian.
  If $\norm{H' - H_F} \leq \epsilon$, then $H(t)$ $\epsilon$-Floquet-simulates $H'$.
\end{corollary}

The following lemma is standard (see e.g.~\cite[Lemma 30]{Cubitt2019}).
\begin{lemma}\label{Lipschitz}
  For $H$, $H'$ Hermitian and $\norm{\cdot}$ any unitarily invariant norm,\linebreak \mbox{$\norm{e^{iH't} - e^{iHt}} \le \norm{H'-H}t$}.
\end{lemma}
This immediately implies (via \cref{def:Floquet_simulation,def:Floquet_Hamiltonian}) that Floquet-simulation reproduces the time-dynamics of the target Hamiltonian:

\begin{corollary}
  If $H(t)$ $(\epsilon,n)$-Floquet-simulates $\Htarg$, then for all $n\in\N^+$
  \begin{equation}
    \norm{U(nT)-e^{-i\Htarg T}} \leq \epsilon T,
  \end{equation}
  where $T$ is the period of $H(t)$ and $U(t)$ is the time-evolution operator generated by it:
  \begin{equation}\label{eq:time_ordered_integral_U}
    U(t)=\mathcal T\exp\left(-i\int_0^tH(s)\,ds\right).
  \end{equation}
\end{corollary}

However, the Floquet Hamiltonian, $H_F$, plays a more important role in Floquet physics than simply being the effective generator of the time-evolution over multiples of the period.
The properties of the Floquet Hamiltonian play a role directly analogous to those of a static, time-independent Hamiltonian.
The eigenvalues of $H_F$ are called the ``quasi-energies'', which are conserved modulo $\hbar\omega$, reflecting the fact that the system can absorb or emit energy in quanta of $\hbar\omega$ from the periodic drive.
The associated eigenvectors $\ket{\psi}$ of $H_F$ are called the ``Floquet states'' (or ``Floquet modes''), and correspond to stable states whose evolution under $H(t)$ has a simple periodic dependence on $t$.
Their structure governs phenomena such as Floquet topological phases, time-crystals, localisation, and many other many-body phenomena in Floquet physics.
See e.g.\ \cite{Shirley65,Grifoni1998,Sambe1973} for previous works and textbooks covering this theory, and \cite{Bukov2015,Eckardt2017,Oka2019} for a selection of more recent reviews covering different aspects of Floquet physics in quantum systems.

From the \cref{def:Floquet_simulation} it follows immediately that the quasienergies and Floquet states of $H(t)$ reproduce the eigenenergies and eigenstates of $\Htarg$:

\begin{corollary}
  Let $H(t)$ be a time-periodic Hamiltonian that $(\epsilon,n)$-Floquet-simulates $\Htarg$.
  Let $\varepsilon_i$ be the quasienergies and $\ket{\varphi_i(t)}$ the corresponding Floquet modes of $H(t)$, and let $\lambda_i$ and $\ket{\psi_i}$ denote the eigenvalues and corresponding eigenstates of $\Htarg$.
  Then
\begin{equation}
	\abs{\varepsilon_i - \lambda_i} \leq \epsilon,
	\qquad
	\norm{\ket{\varphi_i(0)} - \ket{\psi_i}} \leq \epsilon.
\end{equation}
\end{corollary}

Because our definition of Floquet simulation approximates the target Hamiltonian in operator norm, this in turn immediately implies that many of the complexity theoretic results for time-independent Hamiltonians carry over to their Floquet analogues.
E.g.\ approximating the minimum quasienergies to $1/poly$ precision is QMA-hard, approximating the quasienergy spectrum is BQP-complete in the same setting in which spectral approximation for static Hamiltonians is, etc.

It is worth highlighting here the differences between our Floquet simulation results and classic Trotterization-based quantum algorithms for Hamiltonian simulation.
That family of quantum algorithms simulates the time dynamics of a quantum Hamiltonian by constructing a sequence of quantum gates that approximates time-evolution under the target Hamiltonian for a small time-step, and then repeats this series of gates.
If we represent the quantum circuit for such a Trotter algorithm as a time-dependent Hamiltonian, then this Hamiltonian is indeed time-periodic, as the sequence of gates is repeated.
Such circuits are sometimes called ``Floquet circuits''.
But the time dependence in this Hamiltonian is very far from a finite sum of smooth sinusoids; it requires ``bang-bang'' control, i.e.\ the local Hamiltonians must be switched on and off discontinuously, equivalent to square-wave driving whose Fourier series has infinite bandwidth.

\section{Main results}\label{sect:main results}

We will be interested in families of time-periodic Hamiltonians that are composed of sums of individual time-dependent interactions from some finite set.
Analogously to the time-independent case, we define $\mathcal{S}$-driven Hamiltonians as the set of all time-dependent Hamiltonians that can be constructed by periodic driving of interactions in the set $\mathcal{S}$:

\begin{definition}[$\mathcal{S}$-driven Hamiltonian]\label{def:S-driven}
  Let $\mathcal{S}=\{h_1,\dots,h_m\}$ be a subset of interactions (Hermitian matrices, each defined on a fixed number of qudits).
  The set of \emph{$\mathcal{S}$-driven Hamiltonians} comprises all periodic Hamiltonians $H(t)$ on any number of qudits $n$ which can be decomposed in terms of these interactions as
  \begin{equation}
    H(t) = \sum_{x=1}^m f_x(t) h_x,
  \end{equation}
  where the $f_x(t)$ are periodic functions with a common period $T$ and base frequency $\omega = \frac{2\pi}{T}$.
\end{definition}

The first main result of this work is the characterisation of the simulation power of Floquet Hamiltonians (see \cref{thm:bounded_amplitude_universal} for the full theorem statement):

\begin{theorem}[Floquet simulation and Floquet universality]\label{thm:Floquet_universality}
  Let $\Lie(\cS)$ be a perfect Lie algebra, generated by the set $\mathcal{S}$. Then an $\mathcal{S}$-driven {Hamiltonians} can Floquet-simulate any Hamiltonian in $\Lie(\cS)$.
  In particular, an $\mathcal{S}$-Hamiltonian is Floquet-universal iff $\mathcal{S}$ generates the full Lie algebra.

  Moreover, the ratios of the amplitudes of all $f(t)$ in \cref{def:S-driven} are bounded between 1 and~2, and the $f(t)$ have finite Fourier series.
\end{theorem}

Notably, while it is significant that in the above construction all the amplitudes are~$\mc O(1)$ and their ratios are between 1 and~2, in the general case the frequencies have unfavourable scaling with the size of the target Hamiltonian being simulated.
Our second result shows that for an important class of target Hamiltonian, including $k$-local lattice Hamiltonians, this Floquet-simulation can be performed efficiently.
In this setting, in addition to all the interaction strengths and their ratios being $\mc O(1)$, the driving frequencies scale only polynomially in the number of parameters in the target Hamiltonian.

\begin{theorem}[Floquet efficiency]\label{thm:Floquet_efficiency_overview}
  $k$-local Hamiltonians whose interaction graph has constant chromatic number (which includes all lattice Hamiltonians) can be simulated by Floquet-universal Hamiltonians with all amplitude ratios bounded between 1 and~2, as above, and with all frequencies in the Fourier series scaling polynomially with the size of the Hamiltonian being simulated.
\end{theorem}
\noindent This result is shown by careful accounting of the requirements on the driving frequencies in \cref{sect:efficiency}.
A formal statement of the efficiency result is given in \cref{thm:Floquet_efficiency}.

\section{Proof roadmap}
The proof of \cref{thm:Floquet_universality} is somewhat lengthy and is organised into four key steps.
Each step is captured by a separate theorem and developed in its own section of~\cref{part:technical}.

In \cref{sec:Magnus_expansion}, we prove that given certain conditions on the driving frequencies in the periodic Hamiltonian the lower order terms of the Magnus expansion vanish (\cref{thm:vanishing_lower_orders}).
The remaining leading order term consists of a linear combination of nested commutators of terms appearing in the time-periodic Hamiltonian.
This term divides into a product of two distinct parts: an outer sum over permutations of nested commutators of elements, and an inner sum over permutations of combinatorial functions of the driving frequencies (\cref{thm:leading_order_term}).

\Cref{sec:proof:Xi} deals with the second of these two parts.
We essentially show that the combinatorial functions appearing in the sum form a complete basis of rational functions in a suitable space.
The driving frequencies can therefore always be chosen such that the inner sum effectively reduces to the identity, leaving just the outer sum over nested commutators of generators.
Ironically, showing that this combinatorial factor can be simplified away entirely is the most involved part of the proof, culminating in \cref{thm:filtering}.

\Cref{sec:proof:canonical_projection} then shows in \cref{thm:leading_order_lie_polynomial} that the remaining outer sum is in fact the canonical projection of words in the free associative algebra over the generators into the Lie algebra over the same.
This implies that any Lie polynomial of fixed degree generated by the terms in the original Hamiltonian can be realised in the leading order term in the Magnus expansion.

The final step in \cref{sec:proof:Floquet_universality} addresses the challenge of combining Lie-polynomial terms appearing at different orders in the Magnus expansion, while controlling the unwanted cross-terms and the infinite tail of the expansion.
The solution is two-fold.
First, we show in \cref{lem:bounded_amplitude_synthesis} that rescaling of the relative frequencies alone allows us to tune the effective magnitudes of terms in the Magnus expansion, allowing terms of different orders to be combined successfully.
Secondly, recent new error bounds on the Magnus expansion due to some of the authors~\cite{ACO} give sufficiently tight control on the norm of the tail.
This leads to the formal statement of our main Floquet universality result \cref{thm:bounded_amplitude_universal}.

\medskip

To analyse the scaling behaviour of the frequency coefficients for the simulation of an arbitrary target Hamiltonian, \cref{lem:dilated_frequency_scaling} first provides a bound in terms of the number of polynomials and their degree.
We then restrict the setting in \cref{sect:efficiency} by assuming a commutation structure of the target decomposition; in \cref{lem:strengthened_frequency_scaling}, commutation between the generators of distinct terms allows the same frequency vectors to be reused within each family, while constant phase shifts bring contributions of different degrees to a common scale. Consequently, the frequency bound depends on the number of families and the maximal Lie polynomial degree, but is independent of the total number of terms within each family.

For local Hamiltonians, \cref{lem:graph_colouring_decomposition} provides the required decomposition by colouring the intersection graph, so that the number of commuting families is controlled by locality and the chromatic number. In particular, lattice Hamiltonians admit such a decomposition with a constant number of families. Combining this with \cref{lem:strengthened_frequency_scaling} and the bounded amplitude construction above yields \cref{thm:Floquet_efficiency_overview}: $k$-local lattice Hamiltonians, and more generally those with constant chromatic number, can be Floquet-simulated with polynomially bounded driving frequencies while all physical amplitudes remain $\mc O(1)$ and within a ratio of~2.

\subsection*{Acknowledgements}
We would like to thank Robert K\"onig, Thomas Schulte--Herbrüggen and Libor Caha for insightful discussions.
HA is supported by EPSRC DTP Grant Reference: EP/N509577/1 and EP/T517793/1.

\clearpage

\part{Mathematical proofs}\label{part:technical}

\section{Preliminaries}

As usual, we will work in natural units where $\hbar = 1$.
Throughout, $\norm{\bullet}$ denotes the operator norm, whereas the trace norm will be denoted $\norm{\bullet}_1$.

\begin{definition}[Lie polynomial]\label{def:Lie_polynomial}
  A \emph{Lie polynomial} over an alphabet $\cA$ is an element of the free Lie algebra $\Lie(\cA)\subset R\langle\cA\rangle$.
  A Lie polynomial $\LP(A_1,\dots,A_n)$ is \emph{homogeneous of degree $K$} if each term in the polynomial contains exactly $K$ of the $A_j$.
  A homogeneous Lie polynomial $\LP(A_1,\dots,A_K)$ is \emph{multilinear of degree $K$} if each $A_j$ appears exactly once in every monomial, so that
  \begin{equation}
    \LP(A_1,\dots,A_K) = \sum_{\pi\in S_K} F(\pi)A_{\pi(1)}\cdots A_{\pi(K)}
  \end{equation}
  for some coefficients $F:S_K\to R$.
  (Note that the $A_j$ are formal elements of the alphabet $\cA$, so different $A_j$ are allowed to correspond to the same operator.)
\end{definition}

\section{Floquet Hamiltonian and Magnus expansion}\label{sec:Magnus_expansion}

The Magnus expansion~\cite{Magnus1954} provides an exponential representation of the solution of ordinary linear differential equations, in particular the time-dependent Schr\"{o}dinger equation.
Applying this to evolution under a time-dependent Hamiltonian, the evolution operator can be expressed as
\begin{equation}
  U(t) = \exp\left[ \mathcal{M}(t)\right],
\end{equation}
where the exponent is given by an infinite series expansion over terms involving multiple integrals of nested commutators:
\begin{equation}\label{eq:Magnus_basic}
  \mathcal{M}(t) = \sum_{k=1}^\infty M_k(t).
\end{equation}

Refs.~\cite{Arnal2017,bandiera2017,Ebrahimi2025Magnus} give an explicit expression for the $k$'th Magnus term as a sum of right nested commutators:
\begin{equation}\label{eq:original_Magnus}
  M_k(t) =
 \frac{(-i)^k}{k^2} \int_0^t dt_1 \cdots \int_0^{t_{k-1}}dt_k \sum_{\sigma \in S_k}
  \frac{(-1)^{\des(\sigma)}}{\binom{k-1}{\des(\sigma)}}
  \left[ H(t_{\sigma(1)}), H(t_{\sigma(2)}),\dots,H(t_{\sigma(k)})\right],
\end{equation}
where we adopt a standard condensed notation for the right-nested commutator:
\begin{equation}
  [X_1,\dots,X_k] \coloneqq [X_1,[X_2,\dots,[X_{k-1},X_k]\dots]].
\end{equation}
Here $\sigma$ is one of the possible $k!$ permutations of ${1,2,\dots,k}$ and $\des(\sigma)$ is the number of descents\footnote{A descent in a permutation $\sigma$ is any position $i\in{1,2,...,n-1}$ for which $\sigma(i+1)<\sigma(i)$.} in the permutation $\sigma$.
In this work we will sometimes denote by $M_k(H,t)$ the $k$'th order Magnus expansion of a periodic Hamiltonian $H$ evaluated at time $t$.

The Magnus expansion can be used to approximate the Floquet Hamiltonian for a time-periodic system.
From \cref{thm:Floquet-Bloch}, there is a direct correspondence between $H_F$ and the Magnus expansion evaluated at the period:
\begin{equation}
  \mathcal{M}(T) = -i H_F T.
\end{equation}
  The Floquet Hamiltonian can then be written in terms of the Magnus expansion:\footnote{Note the differing indexing of the expansion terms.
    This is to be consistent with the Magnus expansion literature and for convenience when considering the Floquet Hamiltonian.}
  \begin{equation}
    H_F = \sum_{k=0}^\infty H_F^{(k)} \quad\text{ where }\quad H_F^{(k)} \coloneq \frac{i}{T}M_{k+1}(T).
  \end{equation}
By explicitly calculating the dominant orders in this expansion, we can approximate the Floquet Hamiltonian.
An important property of this Magnus expansion is that it remains Hermitian when truncated (unlike e.g.\ the Dyson expansion).

We are concerned here with $\cS$-driven Hamiltonians of the form in \cref{def:S-driven}.
It will be convenient to express the periodic driving functions $f_x(t)$ in terms of their Fourier series.
We will restrict to driving functions with finite Fourier series expansions and all terms of cosine form except (at most) one, which will suffice for our purposes.

\begin{definition}[$\cS$-driven Hamiltonian (Fourier expansion)]\label{def:S-driven_fourier}
  For an $\cS$-driven Hamiltonian from \cref{def:S-driven}, we write
  \begin{equation}
    H(t) =
    \begin{cases}
      \sum_{j=1}^{C} a_j \left(e^{i \ell_j \omega t} + e^{-i \ell_j \omega t} \right) A_j
      & \text{$C$ odd}, \\[2ex]
      \sum_{j=1}^{C-1} a_j \left(e^{i \ell_j \omega t} + e^{-i \ell_j \omega t} \right) A_j
                      + i \,a_C \left(e^{-i \ell_C \omega t} - e^{i \ell_C \omega t} \right) A_C
      & \text{$C$ even}.
    \end{cases}
  \end{equation}
  Here, $\omega>0$ is the base frequency, $\ellvec = \ell_1,\dots,\ell_j,\dots,\ell_C$ is the vector of driving frequencies, $a_j$ are the amplitudes for each frequency component, and the operators $A_j \in \cS$ are chosen from the set $\cS$ of available Hamiltonian terms.
\end{definition}

Compared to \cref{def:S-driven} of $\cS$-driven Hamiltonians, this regroups terms into $C$ different Fourier frequency components.
We can easily recover the form of \cref{def:S-driven}, by setting the periodic driving functions to be
\begin{equation}
  f_x(t) = \sum_m a_m \sum_{j:A_j=h_x} \left(\e^{i \ell_j^{(m)} \omega t} + \e^{-i \ell_j^{(m)} \omega t} \right)
\end{equation}
for odd $C$, and analogously for even $C$.

Here and throughout, we will choose driving frequencies that satisfy certain conditions, described by the following.

\begin{condition}[Admissible frequency vectors]\label{def:admissible_frequency_vector}
  A frequency vector
  \begin{equation}
    \ellvec = (\ell_1,\dots,\ell_K)\in \mathbb Z^K
  \end{equation}
  is called \emph{admissible} if for every
  \(1\leq k\leq K\), every choice of indices $(i_1,\dots,i_k)\in\{1,\dots,K\}$, and every choice of signs $s_1,\dots,s_k\in\{\pm1\}$,
  the relation
  \begin{equation}
    \sum_{q=1}^k s_q\ell_{i_q}=0
  \end{equation}
  holds if and only if one of the following two alternatives is satisfied:
  \begin{enumerate}[label=(\roman*)]\label{cond:zero_sum}
  \item $k=K$, every index is distinct, and
    $s_1=\cdots=s_K=\pm1$. I.e.
    \begin{equation*}
      \sum_{j=1}^K \ell_j=0.
    \end{equation*}
  \item $\{(i_1,s_1),\dots,(i_k,s_k)\}$ admits a pairing such that each pair consists of identical indices with opposite signs,
    \begin{equation}
      (i_a,s_a),(i_b,s_b)\quad\text{with}\quad i_a=i_b,\ s_a=-s_b .
    \end{equation}
  \end{enumerate}
\end{condition}

We will later need to divide the frequencies into ``blocks'' $\ellvec^{(m)}$, chosen such that each block of frequencies contributes independently to the Magnus expansion at leading order.
To this end, we refine \cref{def:S-driven_fourier} to introduce additional block indices.
Note that this is purely notational; it is easy to see that this does not change the class of Hamiltonians under consideration.

\begin{definition}\label{def:S-driven_blocks}
  For an $\cS$-driven Hamiltonian from \cref{def:S-driven_fourier}, we divide the frequencies into multiple blocks $\ellvec^{(m)}$, and write
  \begin{equation}
    H(t) = \sum_m H_m(t)
  \end{equation}
  where
  \begin{adjustwidth}{-1cm}{-1cm}
	  \begin{equation}\label{eq:H_ell_s_m_def}
	    H_m(t) =
	    \begin{cases}
	      \sum_{j=1}^{K} a^{(m)}_j \left(e^{i \ell^{(m)}_j \omega t} + e^{-i \ell^{(m)}_j \omega t} \right) A_j
	      & \text{$K$ odd}, \\[2ex]
	      \sum_{j=1}^{K-1} a^{(m)}_j \left(e^{i \ell^{(m)}_j \omega t} + e^{-i \ell^{(m)}_j \omega t} \right) A_j
	                      + i \,a^{(m)}_K \left(e^{-i \ell^{(m)}_K \omega t} - e^{i \ell^{(m)}_K \omega t} \right) A_K
	      & \text{$K$ even}.
	    \end{cases}
	  \end{equation}
  \end{adjustwidth}
\end{definition}

Note that, here, we deliberately choose to restrict every block $H_m$ to contain the same number $K$ of frequency components.
We extend the notion of admissible frequencies to such sets of frequency vectors $\ellvec^{(m)}$, imposing an additional condition on the interaction between distinct frequency blocks.

\begin{condition}[Admissible frequency sets]\label{def:admissible_frequency_set}
	We say that a set of frequency vectors
	\begin{equation}
		\mathcal{L}=\{\ellvec^{\,(1)},\dots,\ellvec^{\,(d)}\}
	\end{equation}
	is \emph{admissible} if the following two conditions hold:
	\begin{enumerate}[label=(\roman*)]
		\item Each individual vector \(\ellvec^{\,(n)}\in\mathcal{L}\) is admissible in the sense of \cref{def:admissible_frequency_vector}.
		\item Any vanishing signed sum involving components from more than one vector in $\mathcal L$ admits a pairing such that each pair consists of identical labelled components with opposite signs; namely, the terms can be partitioned into pairs of the form
		\begin{equation}
			\left\{\ell_i^{(n)},-\ell_i^{(n)}\right\}.
		\end{equation}
	\end{enumerate}
\end{condition}

\subsection{Vanishing lower orders}

We will show that the above conditions on the driving frequencies lead to vanishing of the lower-order terms of the Magnus expansion. First we prove a useful lemma that is a consequence of the fact that the integral of a non-trivial Fourier mode over a full period vanishes.

\begin{lemma}\label{lem:vanishing_nested_integrals}
  Let \(T=2\pi/\omega\), and let $a_1,\dots,a_K\in\mathbb{Z}$.
  Consider any ordered sequence
  \begin{equation}
    (a_{i_1},\dots,a_{i_J}),\qquad
    1\leq J\leq K,\quad i_1,\dots,i_J\in\{1,\dots,K\},
  \end{equation}
  where repetitions of indices are allowed.
  If no non-empty subset of the selected terms vanishes, i.e.
  \begin{equation}
    \sum_{m\in S}a_{i_m}\neq0,
    \qquad
    \emptyset \neq S\subseteq\{1,\dots,J\},
  \end{equation}
  then the following associated integral vanishes:
  \begin{equation}
  	\int_0^Tdt_1\int_0^{t_1}dt_2\cdots\int_0^{t_{J-1}}dt_J\,
  	\exp\left(i\omega\sum_{m=1}^{J}a_{i_m}t_m\right)
  	=0.
  \end{equation}
\end{lemma}

\begin{proof}
  Denote the integral in the statement by \(I_J(a_{i_1},\dots,a_{i_J})\). Performing the innermost integration over \(t_J\) gives
  \begin{equation}
  	\begin{multlined}
  		I_J(a_{i_1},\dots,a_{i_J})
  		=
  		\frac{1}{i\omega a_{i_J}}
  		\int_0^Tdt_1\int_0^{t_1}dt_2\cdots\int_0^{t_{J-2}}dt_{J-1} \\
  		\times\left(
  		\exp\left[i\omega\left(\sum_{m=1}^{J-2}a_{i_m}t_m+(a_{i_{J-1}}+a_{i_J})t_{J-1}\right)\right]
  		-
  		\exp\left[i\omega\sum_{m=1}^{J-1}a_{i_m}t_m\right]
  		\right).
  	\end{multlined}
  \end{equation}
  The two terms correspond to evaluating the boundary of the innermost integration at $t_J=t_{J-1}$ and $t_J=0$, respectively.
  The integrations are then performed recursively: at each recursive step, the coefficient of the integration variable is a sum of a non-empty subset of the selected terms $\{a_{i_1},\dots,a_{i_J}\}$. By assumption, every such sum is nonzero, so every recursive integration is well-defined and follows the pattern
  \begin{equation}
    \int_0^{t_m} dt_m \, \e^{i \omega \left(\sum_{j=1}^{m-1} a_{i_j} t_j + \phi t_m\right)}
    =
    \frac{1}{i \omega \, \phi} \left(e^{i \omega \left( \sum_{j=1}^{m-2} a_{i_j} t_j + (a_{i_{m-1}} + \phi)t_{m-1} \right) }
      - e^{i \omega \sum_{j=1}^{m-1} a_{i_j} t_j}\right) \ ,
  \end{equation}
  where $\phi$ is a non-zero integer given by a sum of a non-empty subset of the selected terms, determined by the preceding integrations.

  After integrating $J-1$ times, the integral can therefore be written as
  \begin{equation}
    I_J(a_{i_1},\dots,a_{i_J})
    = \frac{1}{(i\omega)^{J-1}}\sum_{c}\alpha_c\int_0^Tdt_1\,\exp(i\omega\gamma_c t_1),
  \end{equation}
  where each $\gamma_c$ is a non-zero integer obtained as a sum of a non-empty subset of the selected terms, and $\alpha_c$ are rational coefficients arising from the recursive integrations, whose explicit values are not needed.

  Finally, since $T=2\pi/\omega$ and $\gamma_c\in\mathbb{Z}\setminus\{0\}$,
  \begin{equation}
    \int_0^Tdt_1\,\exp(i\omega\gamma_c t_1) = \frac{\exp(i2\pi\gamma_c)-1}{i\omega\gamma_c} = 0
  \end{equation}
  Hence the final expression can be written as a sum of terms, each vanishing, resulting in
  \begin{equation}
    I_J(a_{i_1},\dots,a_{i_J}) = 0.
  \end{equation}
\end{proof}

\begin{theorem}\label{thm:vanishing_lower_orders}
  If the driving frequencies $\ellvec^{(m)}$ in \cref{def:S-driven_fourier} satisfy \cref{def:admissible_frequency_set}, then $M_k(T) = 0$ for all $k < K$.
\end{theorem}

\begin{proof}
  Substituting the Fourier expansion from \cref{def:S-driven_fourier} into the $k$-th Magnus term \cref{eq:original_Magnus}, results in a sum of nested commutators weighted by nested oscillatory time integrals:
  \begin{equation}
    M_k(T)
    =
    \sum_{\mathbf n,\mathbf j,\mathbf s,\sigma}
    C(\mathbf n,\mathbf j,\mathbf s,\sigma)
    \left[A_{j_1},\dots,A_{j_k}\right]
    I(\mathbf n,\mathbf j,\mathbf s,\sigma),
  \end{equation}
  where the sum is taken over
  \begin{equation}
    \mathbf n\in\{1,\dots,d\}^k,\qquad
    \mathbf j\in\{1,\dots,K\}^k,\qquad
    \mathbf s\in\{\pm1\}^k,\qquad
    \sigma\in S_{k-1}.
  \end{equation}
  Here $n_r$ specifies the frequency vector $\ell^{(n_r)}$ from which the $r$-th Fourier component is taken, $C(\mathbf n,\mathbf j,\mathbf s,\sigma)$ collects the combinatorial factors and amplitudes and
  \begin{equation}\label{eq:isolated_integral}
    I(\mathbf n,\mathbf j,\mathbf s,\sigma)
    =
    \int_0^Tdt_1\cdots\int_0^{t_{k-1}}dt_k\,
    \exp\left(
      i\omega
      \sum_{r=1}^k
      s_{\sigma^{-1}(r)}
      \ell^{(n_{\sigma^{-1}(r)})}_{j_{\sigma^{-1}(r)}}
      t_r
    \right).
  \end{equation}
  Therefore, to prove the vanishing of \(M_k(T)\), it is sufficient to
  show that every nested oscillatory integral \(I(\mathbf j,\mathbf s,\sigma)\) vanishes
  for \(k<K\).

  For a fixed choice of $(\mathbf n,\mathbf j,\mathbf s,\sigma)$, define
  \begin{equation}
    a_r \coloneqq s_{\sigma^{-1}(r)} \ell^{(n_{\sigma^{-1}(r)})}_{j_{\sigma^{-1}(r)}},
    \qquad r=1,\dots,k.
  \end{equation}
  Then $I(\mathbf n,\mathbf j,\mathbf s,\sigma)$ is precisely the nested oscillatory integral in \cref{lem:vanishing_nested_integrals} associated with the sequence \((a_1,\dots,a_k)\).

  Since the driving frequencies satisfy \cref{def:admissible_frequency_set}, every choice of $I(\mathbf n,\mathbf j,\mathbf s,\sigma)$ satisfies exactly one of the two alternatives in the admissibility condition.
  \begin{enumerate}[label=(\roman*)]
  \item The signed sequence $(a_1,\dots,a_k)$ does not decompose into opposite-sign pairs.
   Since $k<K$, the exceptional $K$-term resonance in \cref{def:admissible_frequency_vector} cannot occur.
   Therefore, if all selected components originate from a single frequency vector, the admissibility of that vector implies that no non-empty subset of the selected terms has zero sum. If the selected components involve more than one frequency vector, the same conclusion follows directly from condition (ii) of \cref{def:admissible_frequency_set}.

    Hence, in either case,
    \begin{equation}
      \sum_{r\in S} a_r\neq0,
      \qquad
      \emptyset\neq S\subseteq\{1,\dots,k\},
    \end{equation}
    and by \cref{lem:vanishing_nested_integrals} the integral is vanishing, $I(\mathbf n,\mathbf j,\mathbf s,\sigma)=0$.
  \item The selected terms admit a decomposition into opposite-sign pairs.
    In this case, the hypothesis of \cref{lem:vanishing_nested_integrals} is not satisfied, since some subsets of the sequence $(a_1,\dots,a_k)$ may have zero sum.
    We therefore show that all potentially non-vanishing contributions cancel pairwise.

    Consider a fixed contribution $\mathcal I$ to the integral $I(\mathbf n,\mathbf j,\mathbf s,\sigma)$ arising from a particular choice
    of boundary evaluations in the recursive integration.
    During the recursive integration, each coefficient generated for an integration variable
    is a sum of a subset of the selected coefficients $a_1,\dots,a_k$.
    If a given integration path contains no vanishing coefficient, then the corresponding contribution vanishes again by
    \cref{lem:vanishing_nested_integrals}.
    Hence only paths containing a zero-sum block can contribute.

    Let $a_n+\cdots+a_m=0$ be a maximal consecutive block whose sum appears as a vanishing coefficient during the recursive integration.
    The maximality fixes the resonant segment uniquely.
    Since the selected terms consist entirely of opposite-sign pairs, this block
    contains an even number of terms.

    For each such contribution $\mathcal I$, define the mirror contribution $\mathcal I'$ with the same $(\mathbf n,\mathbf j,\sigma)$ and flipping the sign vector
    \begin{equation}
      s_r'=
      \begin{cases}
        -s_r,& n\leq r\leq m,\\
        s_r,& \text{otherwise}.
      \end{cases}
    \end{equation}
    $\mathcal I'$ is present in the Magnus sum because the sum contains all sign assignments.
    The opposite-sign pairing condition ensures that this operation defines a unique partner, and the corresponding contribution has exactly the same boundary-evaluation pattern as $\mathcal I$.

    The integrations outside the resonant block are therefore identical for $\mathcal I$ and $\mathcal I'$, while the integrations within the block differ only by the signs of their denominators. The final integration over the zero mode produces only the factor
    \begin{equation}
      \int_0^T dt\,1=T
    \end{equation}
    and does not introduce an additional sign.
    Within the block, there are $m-n$ recursive integrations before the final zero-mode integration and every non-zero coefficient appearing in the recursive integration changes sign, and hence each integration in this block contributes an additional factor $-1$.
    \begin{equation}
      \mathcal I'=(-1)^{m-n}\mathcal I.
    \end{equation}
    Because the block contains an even number of entries, $m-n+1$ is even, so $m-n$ is odd and
    \begin{equation}
      \mathcal I'=-\mathcal I.
    \end{equation}
    Thus every non-zero contribution is paired with a unique mirror contribution of opposite sign and the corresponding terms cancel in the Magnus sum, and therefore the total contribution from case (ii) vanishes.
  \end{enumerate}

  Combining the two cases, every term in the expansion of $M_k(T)$ either vanishes individually or cancels with its mirror partner, yielding $M_k(T)=0$ for $k<K$ as stated.
\end{proof}

\subsection{Leading-order Magnus terms}

The coefficients appearing in the Magnus expansion are conveniently expressed in terms of the following Eulerian weights, which we will use throughout the paper.

\begin{definition}[Eulerian weights]\label{def:eulerian_weights}
  For a permutation $\sigma\in S_K$, let
  \begin{equation}
    \des(\sigma)\coloneqq \#\{j\in\{1,\dots,K-1\}:\sigma(j)>\sigma(j+1)\}.
  \end{equation}
  denote the number of descents in $\sigma$.
  The associated \textup{Eulerian weight} of $\sigma$ is
  \begin{equation}
    w(\sigma)\coloneqq \frac{(-1)^{\des(\sigma)}}{\binom{K-1}{\des(\sigma)}}.
  \end{equation}
\end{definition}

The vanishing of all orders below $K$ leaves the $K$'th Magnus term as the leading contribution, whose explicit form we now derive.

\begin{theorem}\label{thm:leading_order_term}
  Let $H(t)$ be of the form in \cref{def:S-driven_fourier} with a single admissible frequency vector $\ellvec$.
  Then the leading-order Magnus term is:
  \begin{equation}
    M_K(T) =
    (-i)^K(-1)^{\lfloor\frac{K-1}{2}\rfloor}\frac{2T\prod_{j=1}^K a_j}{K^2\omega^{K-1}}
    \sum_{\mu\in S_K}[A_{\mu(1)},\dots,A_{\mu(K)}]
    \sum_{\sigma\in S_K} w(\sigma) \, \Xi(\mu\,\sigma^{-1},\ellvec)
  \end{equation}
  where $w(\sigma)$ is the Eulerian weight from \cref{def:eulerian_weights}, and we define
  \begin{equation}\label{def:Xi_ell}
    \Xi(\pi,\ell) \coloneqq \prod_{J=1}^{K-1}\left(\sum_{i=1}^J \ell_{\pi(i)} \right)^{-1}.
  \end{equation}
\end{theorem}

\begin{proof}

  By \cref{thm:vanishing_lower_orders}, $M_k(T) = 0$ for all $k < K$ and the first potentially non-zero contribution to the Magnus expansion is the order-$K$ term.
  Specialising the expansion in \cref{eq:isolated_integral} to $k=K$ and $d=1$, the relevant oscillatory integrals take the form
  \begin{equation}\label{eq:leading_order_integral}
    I(\mathbf j,\mathbf s,\sigma)
    =
    \int_0^Tdt_1\cdots\int_0^{t_{K-1}}dt_K\,
    \exp\left(\mathrm{i}\omega \sum_{r=1}^{K} s_{\sigma^{-1}(r)} \ell_{j_{\sigma^{-1}(r)}}t_r \right).
  \end{equation}

  By the admissibility condition, any signed sum of fewer than $K$ selected components either contains opposite-sign pairs or is non-resonant.
  The these contributions vanish by the argument in \cref{thm:vanishing_lower_orders}.
  However, at order $K$, the only remaining possibility is the full $K$-frequency resonance of \cref{def:admissible_frequency_vector}(i), where every frequency component appears exactly once with the same sign.

  Hence the remaining contributions $M_K(T)$ correspond to the integrals \cref{eq:leading_order_integral} for which $\{j_1,\dots,j_K\}=\{1,\dots,K\}$ i.e. every frequency coefficient appears exactly once.
  Thus the index vector $\mathbf{j}$ is restricted to a permutation of
  $\{1,\dots,K\}$, and can be relabelled by $\mu\in S_K$.
  The associated nested commutator therefore contains every operator $A_1,\dots,A_K$ once and the $K$'th term is given by
  \begin{equation}\label{eqn:magnus_itermediate}
    M_K(T)
    =
    (-i)^K \frac{\prod_{j=1}^K a_j}{K^2}
    \sum_{\sigma\in S_K} w(\sigma)
    \sum_{\mu\in S_K}
    [A_{\mu(1)},\dots,A_{\mu(K)}]
    \sum_{\mathbf s\in\{\pm1\}^K}
    I(\mu,\mathbf s,\sigma).
  \end{equation}

  The next step is to evaluate $ \sum_{\mathbf s}I(\mu,\mathbf s,\sigma)$ for a fixed commutator ordering $\mu$ and integration ordering $\sigma$.
  In the recursive evaluation of the nested integral, each integration produces two boundary terms, corresponding to evaluating the upper limit ($t_{r-1}$) or the lower limit ($0$).
  Any boundary choice involving a lower limit produces, at some stage, a partial sum over a proper subset of the frequency vector elements as the integration coefficient.
  This is non-zero by admissibility, and hence the corresponding contribution vanishes by \cref{lem:vanishing_nested_integrals}.
  Therefore only the all-upper boundary path survives.

  For this path, the coefficient of the final integration variable $t_1$ is the full signed sum
  \begin{equation}
    \sum_{r=1}^{K}s_{\sigma^{-1}(r)}
    \ell_{\mu (\sigma^{-1}(r))}.
  \end{equation}
  Admissibility allows this sum to vanish only in the case (i) of \cref{def:admissible_frequency_vector} with all signs equal.
  Therefore only the two uniform sign choices $\{+1\}^K$ and $\{-1\}^K$ contribute from the sum over $\mathbf s$.

  For the two surviving $\mathbf s$, the all-upper integrations over
  $t_K,\dots,t_2$ generate a partial sums in the exponent
  \begin{equation}
    \sum_{r=1}^{J}\ell_{\mu(\sigma^{-1}(r))},
    \qquad J=1,\dots,K-1.
  \end{equation}
  with each integration contributing a factor $(\mathrm{i}\omega)^{-1}$ and the inverse of the previous partial sum.
  The remaining $t_1$ integral has zero frequency by the resonance
  condition, and hence contributes a factor of $T$.
  Thus, using \cref{def:Xi_ell},
  \begin{align}
    I\bigl(\mu,\{+1\}^K,\sigma\bigr)
    &=
      \frac{T}{(-i\omega)^{K-1}}
      \Xi(\mu\circ\sigma^{-1},\ellvec),\\
    I\bigl(\mu,\{-1\}^K,\sigma\bigr)
    &=
      \frac{(-1)^{K-1}T}{(-i \omega)^{K-1}}
      \Xi(\mu\circ\sigma^{-1},\ellvec).
  \end{align}

  For odd $K$, the two uniform sign choices contribute with the same sign, since $K-1$ is even contributing a factor of $2$.
  Using $\mathrm{i}^{K-1}=(-1)^{(K-1)/2}$ and substituting into \cref{eqn:magnus_itermediate} gives
  \begin{equation}
    M_K(T)
    =
    (-i)^K(-1)^{\frac{K-1}{2}}
    \frac{2T\prod_{j=1}^K a_j}{K^2\omega^{K-1}}
    \sum_{\mu\in S_K}[A_{\mu(1)},\dots,A_{\mu(K)}]
    \sum_{\sigma\in S_K}
    w(\sigma)\,
    \Xi(\mu\circ\sigma^{-1},\ellvec),
  \end{equation}

  For even $K$, the two uniform sign choices acquire opposite signs from the $K-1$ integrations.
  However, the $A_K$ Fourier component in \cref{def:S-driven_fourier} has coefficients $-\mathrm{i}a_K$ and $+\mathrm{i}a_K$ for the positive- and negative-frequency components, respectively.
  Thus flipping all signs introduces an additional factor of $-1$, which cancels the factor $(-1)^{K-1}=-1$ from the integrations. The two contributions therefore again add, giving a factor of $2$.
  Since $K-1$ is odd,
  \begin{equation}
    \mathrm{i}^{-(K-1)}
    =
    (-1)^{(K-2)/2}(-\mathrm{i}),
  \end{equation}
  and the additional factor of $\mathrm{i}$ from the $A_K$ Fourier coefficient gives the overall sign $(-1)^{(K-2)/2}$.

  Since
  \begin{equation}
    (-1)^{\lfloor (K-1)/2\rfloor}
    =
    \begin{cases}
      (-1)^{(K-1)/2}, & K\ \mathrm{odd},\\
      (-1)^{(K-2)/2}, & K\ \mathrm{even},
    \end{cases}
  \end{equation}
  both cases are captured by the prefactor in the theorem statement.
\end{proof}

The result for an admissible frequency vector immediately extends to an admissible frequency set: condition~(ii) of \cref{def:admissible_frequency_set} ensures that any contribution involving components from different frequency vectors has a non-zero coefficient in the exponential.
Hence, by \cref{lem:vanishing_nested_integrals}, these contributions vanish upon integration over one period, and the leading-order term in the Magnus expansion is simply the sum of the leading-order terms for each vector.

\begin{corollary}\label{cor:leading_order_Magnus_K}
  Let $H(t)=\sum_{m=1}^d H_m(t)$,where each $H_m(t)$ is of the form in \cref{def:S-driven_fourier} with frequency vectors $\{\ellvec^{(1)},\dots,\ellvec^{(d)}\}$ that form an admissible frequency set.
  Then the leading-order Magnus term is
  \begin{equation}
    M_K(T)=\sum_{m=1}^d M_K^{(m)}(T),
  \end{equation}
  where $M_K^{(m)}(T)$ is the leading-order Magnus term associated with $H_m(t)$, given by \cref{thm:leading_order_term}.
\end{corollary}

\begin{proof}
  By \cref{thm:vanishing_lower_orders}, $M_k(T)=0$ for all $k<K$.

  At order $K$, consider a contribution involving Fourier components from more than one frequency vector.
  Each coefficient arising during the recursive integration is a signed sum of a subset of these components.
  If this sum does not vanish, then the corresponding integral vanishes by \cref{lem:vanishing_nested_integrals}.
  If it does vanish, then by \cref{def:admissible_frequency_set} it consists only of pairwise cancellations.
  The corresponding integral therefore vanishes by the mirror argument in \cref{thm:vanishing_lower_orders}.

  Thus, the only non-vanishing order-$K$ contributions involve components from a single frequency vector $\ellvec^{(m)}$.
  By \cref{thm:leading_order_term}, these contributions sum to $M_K^{(m)}(T)$, and therefore
  \begin{equation}
    M_K(T)=\sum_{m=1}^d M_K^{(m)}(T).
  \end{equation}
\end{proof}

\section{Algebraic reduction via residues}
\label{sec:proof:Xi}

We want to identify the space spanned by the rational functions $\{\Xi(\pi,\cdot)\}_{\pi\in S_K}$ from \cref{thm:leading_order_term}.
We do this by analysing their simple poles along subset-sum hyperplanes, because residues isolate the contribution of specific prefix/suffix factors and allow us to establish the image space.\footnote{The technique of isolating coefficients via residues is analogous to the study of the Orlik-Solomon algebra, where the \emph{Leray residue} of a logarithmic form along a hyperplane corresponds to the algebraic restriction operation in Ref.~\cite[Sections~3.1 and~3.2]{OrlikTerao1992}.}

From \cref{def:Xi_ell}, we can write $\Xi(\pi,\ell)$ in the form
\begin{equation}\label{def:Xi_S}
	\Xi(\pi,\ell) \coloneqq \prod_{k=1}^{K-1}\frac{1}{S_k^\pi(\ell)},
	\qquad
	S_k^\pi(\ell):=\sum_{j=1}^k \ell_{\pi(j)},
	\qquad \pi\in S_K,
\end{equation}
viewed as a rational function on the affine hyperplane $H_0$ restricted to the Zariski-open set $U\subset H_0$ where all partial sums $S_k^\pi(\ell)$ are non-zero.
As these sets will be important throughout this section, we capture them in the following definition.

\begin{definition}[$H_0^{(K)}$ hyperplane and Zariski-open set $U$]\label{def:H0_U}
	We define the zero-sum hyperplane
	\begin{equation}\label{eq:H0_def}
		H_0^{(K)}\coloneqq\left\{\ell\in\mathbb C^K:\sum_{i=1}^K\ell_i=0\right\}.
	\end{equation}
	We denote by $U \subset H_0^{(K)}$ the Zariski-open subset on which all partial sums appearing in the denominators of $\Xi$ are non-zero:
	\begin{equation}
		U \coloneqq\left\{\ell\in H_0^{(K)}:S_k^\pi(\ell)\neq0\ \text{for all }\pi\in S_K,\ k=1,\dots,K-1\right\},
	\end{equation}
	where
	\begin{equation}
		S_k^\pi(\ell)\coloneqq\sum_{j=1}^k\ell_{\pi(j)}.
	\end{equation}
\end{definition}

\subsection{Hyperplanes and residues}

In preparation for factorising $\Xi$ via residues, we first formalise the relevant boundary hyperplanes of $H_0$ and the corresponding notion of a simple pole.
\begin{definition}[Subset-sum hyperplanes]\label{def:subset_hyperplanes}
	For a non-empty proper subset $I\subset [K]$ of size $r \coloneqq |I|$, we define the linear form $g_I(\ell) \coloneqq \sum_{i\in I}\ell_i$ and its associated \emph{subset-sum hyperplane}
	\begin{equation}
		H_I \coloneqq \big\{\ell\in H_0:\ g_I(\ell)=0\big\}.
	\end{equation}
	Because every point on $H_0$ satisfies $\sum_{i=1}^K \ell_i = 0$, we immediately have $g_{I^c}(\ell)=-g_I(\ell)$, and therefore $H_I = H_{I^c}$.
\end{definition}

\noindent With the geometry established, we say a rational function $F(\ell)$ on $H_0$ has at most a \emph{simple pole} along $H_I$ if, in some linear coordinates $(u,v)$ on $H_0$ with $u=g_I(\ell)$, one can write
\begin{equation}
	F(u,v)=\frac{A(v)}{u}+B(u,v),
\end{equation}
where $A$ and $B$ are regular near $u=0$ (i.e., have no poles in the neighbourhood). We then define the \emph{residue of $F$ along $H_I$} as the function $A(v)$ evaluated on $H_I$:
\begin{equation}
	\mathrm{Res}_{H_I}(F) \coloneqq A(v)
	\qquad\text{(equivalently, } \mathrm{Res}_{H_I}(F)=\lim_{u\to 0} u\,F(u,v)\text{)}.
\end{equation}

For any permutation $\pi\in S_K$ and $1 \le k \le K-1$, we denote its \emph{prefix set} of length $k$ by
\begin{equation}
	P_k(\pi) \coloneqq \{\pi(1),\dots,\pi(k)\}.
\end{equation}
The remaining elements $\{\pi(k+1), \dots, \pi(K)\}$ form the complementary suffix. Note that for a subset $I$ of size $r$, the condition $P_{K-r}(\pi) = I^c$ is equivalent to the requirement that the suffix of $\pi$ is exactly the set $I$.

\begin{lemma}\label{lem:pole_support}
	Fix a proper subset $I\subset [K]$ of size $r \coloneqq |I|$. The function $\Xi(\pi,\ell)$ has a simple pole along $H_I$ if and only if
	\begin{equation}
		P_r(\pi)=I \qquad \text{(prefix case)}
		\qquad\text{or}\qquad
		P_{K-r}(\pi)=I^c \qquad \text{(suffix case)}.
	\end{equation}
	In the prefix case, the pole arises from the factor $1/S_r^\pi(\ell)$, where $S_r^\pi(\ell)=g_I(\ell)$; in the suffix case, the pole comes from the denominator factor $S_{K-r}^\pi(\ell)$ where $S_{K-r}^\pi(\ell)=g_{I^c}(\ell)=-g_I(\ell)$ on $H_0$.
\end{lemma}

\begin{proof}
	A pole along $H_I$ occurs if and only if one of the partial sums $S_k^\pi(\ell)$ equals $g_I(\ell)$ or $g_{I^c}(\ell)$ on $H_0$. Since $|I|=r$ and $|I^c|=K-r$, this can happen only in the two cases $P_r(\pi)=I$ or $P_{K-r}(\pi)=I^c$.
	In the first case, $P_r(\pi)=I$, we have $S_r^\pi(\ell)=g_I(\ell)$, so the pole comes from the factor $1/S_r^\pi(\ell)$. In the second case, $P_{K-r}(\pi)=I^c$, it holds $S_{K-r}^\pi(\ell)=g_{I^c}(\ell)=-g_I(\ell)$ on $H_0$, so the pole comes from the factor $1/S_{K-r}^\pi(\ell)$.
	In either instance, exactly one partial sum vanishes generically on $H_I$ while all others remain non-zero, hence the pole is simple.
	Conversely, if neither condition holds, no denominator factor equals $g_I(\ell)$ or $g_{I^c}(\ell)$ on $H_0$; thus, none vanishes generically on $H_I$ and no pole is present.
\end{proof}

\medskip

When $P_r(\pi)=I$, we can write $\pi$ uniquely as a concatenation
\begin{equation}
	\pi=(\alpha,\beta),
	\qquad
	\alpha\in S(I)\cong S_r,
	\quad
	\beta\in S(I^c)\cong S_{K-r},
\end{equation}
meaning that the first $r$ values of $\pi$ are the sequence $\alpha(1),\dots,\alpha(r)$ (a permutation of $I$), followed by $\beta(1),\dots,\beta(K-r)$ (a permutation of $I^c$).
For instance, if $K=5$ with $I=\{2,4\}$ and $I^c=\{1,3,5\}$, the permutation $\pi=(4,2,5,1,3)$ has prefix $P_2(\pi)=I$ and splits cleanly into $\alpha=(4,2)$ and $\beta=(5,1,3)$.

From \cref{def:Xi_S}, we define the lower-degree functions for these subsystems:
\begin{equation}
	\Xi_I(\alpha,\ell) \coloneqq \prod_{k=1}^{r-1}\frac{1}{\sum_{j=1}^k \ell_{\alpha(j)}},
	\qquad
	\Xi_{I^c}(\beta,\ell) \coloneqq \prod_{s=1}^{K-r-1}\frac{1}{\sum_{j=1}^s \ell_{\beta(j)}}.
\end{equation}

\begin{lemma}\label{lem:res_factor_complete}
	Fix a proper subset $I\subset[K]$ of size $r \coloneqq |I|$.

	\smallskip

	\noindent\emph{(Prefix case)}
	If $P_r(\pi)=I$ and $\pi=(\alpha,\beta)$, then $\Xi(\pi,\ell)$ has a simple pole along $H_I$ and
	\begin{equation}
		\mathrm{Res}_{H_I}\big(\Xi(\pi,\ell)\big)
		=
		\Xi_I(\alpha,\ell)\cdot \Xi_{I^c}(\beta,\ell).
	\end{equation}

	\smallskip

	\noindent\emph{(Suffix case)}
	If $P_{K-r}(\pi)=I^c$ and $\pi=(\beta,\alpha)$, then $\Xi(\pi,\ell)$ has a simple pole along $H_I$ and
	\begin{equation}
		\mathrm{Res}_{H_I}\big(\Xi(\pi,\ell)\big)
		=
		-\,\Xi_I(\alpha,\ell)\cdot \Xi_{I^c}(\beta,\ell).
	\end{equation}
\end{lemma}

\begin{proof}
	We detail the prefix case; the suffix case follows identically, except that the pole arises from $S_{K-r}^\pi(\ell)=g_{I^c}(\ell)=-g_I(\ell)$ on $H_0$, which contributes the overall minus sign.

	Assume $P_r(\pi)=I$ and $\pi=(\alpha,\beta)$. By \cref{lem:pole_support}, the pole arises solely from the factor $1/S_r^\pi(\ell)$, where $S_r^\pi(\ell)=g_I(\ell)$.
	Isolating this singularity, we write
	\begin{equation}\label{eq:singularity_factorisation}
		g_I(\ell)\,\Xi(\pi,\ell)
		=
		\left( \prod_{k=1}^{r-1}\frac{1}{S_k^\pi(\ell)} \right)
		\cdot
		\left( \prod_{s=1}^{K-r-1}\frac{1}{S_{r+s}^\pi(\ell)} \right).
	\end{equation}
	We now restrict this expression to $H_I$, where $g_I(\ell)=0$.

	For the first product ($k<r$), the partial sum $S_k^\pi(\ell)$ involves only elements from $I$. It is therefore identical to $\sum_{j=1}^k \ell_{\alpha(j)}$, meaning the first product evaluates exactly to $\Xi_I(\alpha,\ell)$.

	For the second product ($s\ge 1$), the partial sum extends into the suffix
	\begin{equation}
		S_{r+s}^\pi(\ell)
		=
		S_r^\pi(\ell)+\sum_{j=1}^s \ell_{\beta(j)}
		=
		g_I(\ell)+\sum_{j=1}^s \ell_{\beta(j)}.
	\end{equation}
	On $H_I$, the $g_I(\ell)$ term vanishes, leaving simply $\sum_{j=1}^s \ell_{\beta(j)}$. Thus, the second product in \cref{eq:singularity_factorisation} evaluates exactly to $\Xi_{I^c}(\beta,\ell)$.

	Multiplying these evaluated products together yields the residue
	\begin{equation}
		\mathrm{Res}_{H_I}\big(\Xi(\pi,\ell)\big)
		=
		\left.\big(g_I(\ell)\,\Xi(\pi,\ell)\big)\right|_{H_I}
		=
		\Xi_I(\alpha,\ell)\cdot\Xi_{I^c}(\beta,\ell).
	\end{equation}
\end{proof}

\begin{example}[Factorisation for $K=5$]
	Consider the hyperplane $H_0=\{\ell_1+\ell_2+\ell_3+\ell_4+\ell_5=0\}$ and define the subset
	\begin{equation}
		I=\{2,4\},\qquad r \coloneqq |I|=2,\qquad I^c=\{1,3,5\}.
	\end{equation}
	The associated pole hyperplane $H_I$ is defined by
	\begin{equation}
		H_I=\{\ell\in H_0 \mid \ell_2+\ell_4=0\}.
	\end{equation}

	\smallskip

	\noindent\emph{Prefix case:}
	Consider the permutation $\pi=(4,2,5,1,3)$, which satisfies $P_2(\pi)=\{4,2\}=I$. We factorise $\pi=(\alpha,\beta)$ as
	\begin{equation}
		\alpha=(4,2)\in S(I) \qquad \text{and} \qquad \beta=(5,1,3)\in S(I^c).
	\end{equation}
	The partial sums are
	\begin{equation}
		S_1^\pi=\ell_4,\qquad S_2^\pi=\ell_4+\ell_2=g_I,\qquad
		S_3^\pi=\ell_4+\ell_2+\ell_5,\qquad S_4^\pi=\ell_4+\ell_2+\ell_5+\ell_1.
	\end{equation}
	Hence, the integration factor is
	\begin{equation}
		\Xi(\pi,\ell) = \frac{1}{\ell_4}\cdot\frac{1}{\ell_4+\ell_2}\cdot\frac{1}{\ell_4+\ell_2+\ell_5}\cdot\frac{1}{\ell_4+\ell_2+\ell_5+\ell_1}.
	\end{equation}
	Multiplying by $g_I(\ell)=\ell_2+\ell_4$ and restricting to $H_I$ (where $\ell_2+\ell_4=0$) yields
	\begin{equation}
		\mathrm{Res}_{H_I}\big(\Xi(\pi,\ell)\big)
		=\left.\big(g_I(\ell)\,\Xi(\pi,\ell)\big)\right|_{H_I}
		=\frac{1}{\ell_4}\cdot\frac{1}{\ell_5}\cdot\frac{1}{\ell_5+\ell_1}.
	\end{equation}
	Observing the subsystem definitions
	\begin{equation}
		\Xi_I(\alpha,\ell)=\frac{1}{\ell_4},
		\qquad
		\Xi_{I^c}(\beta,\ell)=\frac{1}{\ell_5}\cdot\frac{1}{\ell_5+\ell_1},
	\end{equation}
	we see that indeed $\mathrm{Res}_{H_I}\big(\Xi(\pi,\ell)\big) = \Xi_I(\alpha,\ell)\cdot\Xi_{I^c}(\beta,\ell)$.

	\smallskip

	\noindent\emph{Suffix case:}
	Consider the permutation $\pi'=(5,1,3,4,2)$, which satisfies $P_{3}(\pi')=I^c$. We factorise $\pi'=(\beta,\alpha)$ as
	\begin{equation}
		\beta=(5,1,3)\in S(I^c) \qquad \text{and} \qquad\alpha=(4,2)\in S(I).
	\end{equation}
	With $g_{I^c}(\ell) = \ell_1 + \ell_3 + \ell_5$, the relevant partial sums are
	\begin{equation}
		S_1^{\pi'}=\ell_5,\qquad
		S_2^{\pi'}=\ell_5+\ell_1,\qquad
		S_3^{\pi'}=g_{I^c},\qquad
		S_4^{\pi'}=g_{I^c}+\ell_4.
	\end{equation}
	Hence,
	\begin{equation}
		\Xi(\pi',\ell)
		=\Bigl(\frac{1}{\ell_5}\cdot\frac{1}{\ell_5+\ell_1}\Bigr)\cdot \frac{1}{g_{I^c}(\ell)}\cdot\frac{1}{g_{I^c}(\ell)+\ell_4}.
	\end{equation}
	Multiplying by $g_I(\ell)$ and restricting to $H_I$ (using $g_{I^c}(\ell)=-g_I(\ell)$ on $H_0$ and $g_{I^c}(\ell)=0$ on $H_I$) gives
	\begin{equation}
		\mathrm{Res}_{H_I}\bigl(\Xi(\pi',\ell)\bigr)
		=\left.\big(g_I(\ell)\,\Xi(\pi',\ell)\big)\right|_{H_I}
		=-\Bigl(\frac{1}{\ell_5}\cdot\frac{1}{\ell_5+\ell_1}\Bigr)\cdot\frac{1}{\ell_4}.
	\end{equation}
	Since $\Xi_{I^c}(\beta,\ell)=\frac{1}{\ell_5}\cdot\frac{1}{\ell_5+\ell_1}$ and $\Xi_I(\alpha,\ell)=\frac{1}{\ell_4}$, we obtain
	\begin{equation}
		\mathrm{Res}_{H_I}\bigl(\Xi(\pi',\ell)\bigr) = -\,\Xi_I(\alpha,\ell)\cdot\Xi_{I^c}(\beta,\ell),
	\end{equation}
	as claimed in \cref{lem:res_factor_complete}.
\end{example}

\subsection{Shuffle identities and image of the functional $\Xi$}

We now turn from the local pole structure of individual denominators to the global linear relations among the family $\{\Xi(\pi,\cdot)\}_{\pi\in S_K}$. The residue machinery above will later be used to prove that a fixed-last-letter subfamily is linearly independent. Before doing so, we identify the universal relations that the full family satisfies. These are governed by shuffle identities: interleaving two ordered subwords gives a partial-fraction identity, which vanishes on the zero-sum hyperplane $H_0$.

\begin{definition}[Shuffle product]\label{def:shuffle}
	Let $\alpha$ and $\beta$ be two disjoint, non-empty ordered words whose union contains exactly the letters of
	\begin{equation}
		[K]\coloneqq\{1,2,\dots,K\}.
	\end{equation}
	The shuffle set $\operatorname{Sh}(\alpha,\beta)$ is the set of all permutations in $S_K$ obtained by interleaving the letters of $\alpha$ and $\beta$ while preserving the internal order of the letters within $\alpha$ and within $\beta$.
\end{definition}

The following results proves a shuffling relation for $\Phi(\alpha,\beta)$. We note that our case is a special instance of the more general Chen's shuffle identity~\cite[Eq.~(1.5.6)]{Chen1977}.

\begin{lemma}\label{lem:phi_shuffle}
	For an ordered word $w=(w_1,\dots,w_m)$, define the \emph{ordered denominator product}
	\begin{equation}
		\Phi(w,\ell)
		\coloneqq
		\prod_{q=1}^{m}
		\left(\sum_{j=1}^{q}\ell_{w_j}\right)^{-1},
	\end{equation}
	with the convention $\Phi(\emptyset,\ell)=1$.
	Then, for any two disjoint ordered words $\alpha$ and $\beta$,
	\begin{equation}
		\Phi(\alpha,\ell)\Phi(\beta,\ell)
		=
		\sum_{\rho\in\operatorname{Sh}(\alpha,\beta)}\Phi(\rho,\ell).
	\end{equation}
\end{lemma}

Note that if $w=\alpha$ is supported on a subset $I$, then
\begin{equation}\label{eq:Phi_Xi_relation}
	\Phi(\alpha,\ell)=\frac{\Xi_I(\alpha,\ell)}{g_I(\ell)},
\end{equation}
because $\Phi$ includes the final denominator $g_I(\ell)$ whereas $\Xi_I$ omits it.

\begin{proof}
	We prove the identity by induction on $|\alpha|+|\beta|$, with the convention
	\begin{equation}
		\Phi(\emptyset,\ell)=1.
	\end{equation}
	The claim is immediate if one of the two words is empty.
	Assume now that both words are non-empty. Write
	\begin{equation}
		\alpha=(\alpha',a),
		\qquad
		\beta=(\beta',b),
	\end{equation}
	where $a$ and $b$ are the final letters of $\alpha$ and $\beta$, respectively. Set
	\begin{equation}
		A\coloneqq \sum_{i\in\alpha}\ell_i,
		\qquad
		B\coloneqq \sum_{i\in\beta}\ell_i.
	\end{equation}
	Every shuffle of $\alpha$ and $\beta$ either ends in $a$ or ends in $b$. Hence
	\begin{equation}
		\sum_{\rho\in\operatorname{Sh}(\alpha,\beta)}\Phi(\rho,\ell)
		=
		\frac{1}{A+B}
		\left(
		\sum_{\sigma\in\operatorname{Sh}(\alpha',\beta)}\Phi(\sigma,\ell)
		+
		\sum_{\tau\in\operatorname{Sh}(\alpha,\beta')}\Phi(\tau,\ell)
		\right).
	\end{equation}
	By the induction hypothesis,
	\begin{equation}
		\sum_{\sigma\in\operatorname{Sh}(\alpha',\beta)}\Phi(\sigma,\ell)
		=
		\Phi(\alpha',\ell)\Phi(\beta,\ell),
	\end{equation}
	and
	\begin{equation}
		\sum_{\tau\in\operatorname{Sh}(\alpha,\beta')}\Phi(\tau,\ell)
		=
		\Phi(\alpha,\ell)\Phi(\beta',\ell),
	\end{equation}
	and
	\begin{equation}
		\Phi(\alpha,\ell)=\frac{1}{A}\Phi(\alpha',\ell),
		\qquad
		\Phi(\beta,\ell)=\frac{1}{B}\Phi(\beta',\ell).
	\end{equation}
	This yields to the identity
	\begin{align}
		\sum_{\rho\in\operatorname{Sh}(\alpha,\beta)}\Phi(\rho,\ell)
		&=
		\frac{1}{A+B}\left(
		\Phi(\alpha',\ell)\Phi(\beta,\ell)
		+
		\Phi(\alpha,\ell)\Phi(\beta',\ell)
		\right)\\
		&=
		\frac{1}{A+B}\left(\frac{1}{B}+\frac{1}{A}\right)
		\Phi(\alpha',\ell)\Phi(\beta',\ell)\\
		&=
		\frac{1}{AB}\Phi(\alpha',\ell)\Phi(\beta',\ell)\\
		&=
		\Phi(\alpha,\ell)\Phi(\beta,\ell).
	\end{align}
\end{proof}

\begin{lemma}\label{lem:shuffle_xi}
	For any partition of the letters $[K]$ into two non-empty ordered words $\alpha$ and $\beta$, the Floquet denominators satisfy the shuffle relation
	\begin{equation}
		\sum_{\rho\in\operatorname{Sh}(\alpha,\beta)}
		\Xi(\rho,\ell)
		\equiv 0
		\qquad
		(\ell\in H_0).
	\end{equation}
\end{lemma}

\begin{proof}
	Let
	\begin{equation}
		S(\ell)\coloneqq \sum_{i=1}^K \ell_i,
		\qquad
		H_0\coloneqq\{S(\ell)=0\}.
	\end{equation}
	We compare the possible simple pole along the hyperplane $S(\ell)=0$ near a point where the remaining denominators are non-zero. Thus every rational function with at most a simple pole along $H_0$ has a Laurent expansion around it, whose principal part is its $S(\ell)^{-1}$ coefficient.

	By \cref{lem:phi_shuffle}, we have
	\begin{equation}\label{eq:phi_shuffle}
		\Phi(\alpha,\ell)\Phi(\beta,\ell)=\sum_{\rho\in\operatorname{Sh}(\alpha,\beta)}\Phi(\rho,\ell).
	\end{equation}
	Recalling \cref{eq:Phi_Xi_relation}, for every full word $\rho\in S_K$ the final denominator in $\Phi(\rho,\ell)$ is the total sum $S(\ell)$ independent of $\rho$. Hence summing over the shuffle identities $\rho\in\operatorname{Sh}(\alpha,\beta)$ gives
	\begin{equation}\label{eq:S(ell)_pole_expansion}
		\sum_{\rho\in\operatorname{Sh}(\alpha,\beta)}\Phi(\rho,\ell)=S(\ell)^{-1}\sum_{\rho\in\operatorname{Sh}(\alpha,\beta)} \Xi(\rho,\ell),
	\end{equation}
	meaning that the coefficient of the pole $S(\ell)^{-1}$ on the right-hand side of \cref{eq:phi_shuffle} is
	\begin{equation}
		\sum_{\rho\in\operatorname{Sh}(\alpha,\beta)}\Xi(\rho,\ell) .
	\end{equation}

	On the other hand, $\alpha$ and $\beta$ are proper subwords of $[K]$. Therefore none of the denominators in $\Phi(\alpha,\ell)\Phi(\beta,\ell)$ identifies with $S(\ell)$ and thus the left-hand side of \cref{eq:phi_shuffle} has no $S(\ell)^{-1}$ pole along $S(\ell)=0$.
	Since both sides of \cref{eq:phi_shuffle} agree as rational functions for $S(\ell)\neq0$, their Laurent expansions near $S(\ell)=0$ have the same pole part. Although the rational functions themselves need not be defined at $S(\ell)=0$, the coefficient of $S(\ell)^{-1}$ is still a well-defined function after restricting to $H_0$. The coefficient on the left-hand side is zero, and therefore \cref{eq:S(ell)_pole_expansion} gives
	\begin{equation}
		\sum_{\rho\in\operatorname{Sh}(\alpha,\beta)}\Xi(\rho,\ell)
		\equiv 0
		\qquad(\ell\in H_0).
	\end{equation}
\end{proof}

\begin{example}[A four-letter shuffle relation]\label{ex:shuffle_1234}
	Let $K=4$ and consider the shuffle decomposition $\alpha=(2),\ \beta=(1,3,4)$.
	Then
	\begin{equation}
		\operatorname{Sh}(\alpha,\beta)
		=
		\{2134,\ 1234,\ 1324,\ 1342\}
	\end{equation}
	and \cref{lem:shuffle_xi} gives
	\begin{equation}
		\Xi(2134,\ell)+\Xi(1234,\ell)+\Xi(1324,\ell)+\Xi(1342,\ell)
		\equiv 0
		\qquad(\ell\in H_0);
	\end{equation}
	indeed, grouping the first two terms results in
	\begin{equation}
		\Xi(1234,\ell)+\Xi(2134,\ell)
		=
		\frac{1}{(\ell_1+\ell_2)(\ell_1+\ell_2+\ell_3)}
		\left(\frac{1}{\ell_1}+\frac{1}{\ell_2}\right)
		=
		\frac{1}{\ell_1\ell_2(\ell_1+\ell_2+\ell_3)}.
	\end{equation}
	On $H_0$, this is
	\begin{equation}
		\Xi(1234,\ell)+\Xi(2134,\ell)
		=
		-\frac{1}{\ell_1\ell_2\ell_4}
	\end{equation}
	and similarly,
	\begin{equation}
		\Xi(1324,\ell)+\Xi(1342,\ell)
		=
		\frac{1}{\ell_1(\ell_1+\ell_3)}
		\left(
		\frac{1}{\ell_1+\ell_2+\ell_3}
		+
		\frac{1}{\ell_1+\ell_3+\ell_4}
		\right).
	\end{equation}
	Using $\ell_1+\ell_2+\ell_3=-\ell_4$ and $\ell_1+\ell_3+\ell_4=-\ell_2$ on $H_0$, this becomes
	\begin{equation}
		\Xi(1324,\ell)+\Xi(1342,\ell)
		=
		\frac{1}{\ell_1\ell_2\ell_4}.
	\end{equation}
	Hence the two grouped contributions cancel, as predicted by the shuffling relation.
\end{example}

\begin{lemma}\label{lem:independent_shuffle_relations}
	Let
	\begin{equation}
		W_K\coloneqq \operatorname{span}_{\mathbb C}\{E_\pi:\pi\in S_K\}
	\end{equation}
	be the formal vector space generated by permutations, where $E_\pi$ is an independent basis vector indexed by $\pi$. For every permutation $\pi\in S_K$ not ending in $K$, write uniquely
	\begin{equation}
		\pi=(u,K,v),
		\qquad v\neq \emptyset,
	\end{equation}
	and define the formal shuffle relation
	\begin{equation}
		R_\pi
		\coloneqq
		\sum_{\rho\in\operatorname{Sh}((u,K),v)}E_\rho.
	\end{equation}
	Then the family
	\begin{equation}\label{eq:family_no_K}
		\{R_\pi:\pi\in S_K,\ \pi \text{ does not end in } K\}
	\end{equation}
	is linearly independent. In particular, the shuffle relations contain
	\begin{equation}
		K!-(K-1)!
	\end{equation}
	linearly independent elements.
\end{lemma}

\begin{proof}
	The proof is triangular with respect to the position of the letter $K$: each $R_\pi$ contains the distinguished term $E_\pi$, while all other terms move $K$ strictly to the right. We will show that this structure forces linear independence.

	For a permutation $\eta\in S_K$, let $\operatorname{pos}_K(\eta)$ denote the position of the letter $K$ in $\eta$.
	By construction, the relation $R_\pi$ contains the term $E_\pi$ with coefficient~$1$; this is the shuffle in which all letters of $v$ are placed after the block $(u,K)$: we call $E_\pi$ the \emph{distinguished} term of $R_\pi$.
	Every other shuffle in $\operatorname{Sh}((u,K),v)$ places at least one letter of $v$ before $K$, hence every non-distinguished term $E_\rho$ appearing in $R_\pi$ satisfies
	\begin{equation}
		\operatorname{pos}_K(\rho)>\operatorname{pos}_K(\pi).
	\end{equation}

	Suppose that there is a linear dependence among the family \cref{eq:family_no_K}
	\begin{equation}\label{eq:dependence_R_pi}
		\sum_{\pi:\,\pi\text{ does not end in }K} a_\pi R_\pi=0.
	\end{equation}
	If not all $a_\pi$ vanish, choose $\pi_0$ among those with $a_{\pi_0}\neq 0$ such that $\operatorname{pos}_K(\pi_0)$ is minimal.

	We claim that the coefficient of $E_{\pi_0}$ in \cref{eq:dependence_R_pi} is exactly $a_{\pi_0}$.
	Indeed, $E_{\pi_0}$ appears with coefficient $1$ in $R_{\pi_0}$, as its distinguished term. If $\operatorname{pos}_K(\pi)\geq \operatorname{pos}_K(\pi_0)$ and $\pi\neq\pi_0$, then $E_{\pi_0}$ does not appear in $R_\pi$: it is not the distinguished term, and every non-distinguished term moves $K$ strictly to the right.
	On the other hand, all terms with $\operatorname{pos}_K(\pi)<\operatorname{pos}_K(\pi_0)$ have coefficient $a_\pi=0$ by the minimal choice of $\pi_0$, so they do not contribute to the coefficient of $E_{\pi_0}$.

	Thus as claimed the coefficient of $E_{\pi_0}$ in the dependence is $a_{\pi_0}$. Since the relation \cref{eq:dependence_R_pi} is zero in $W_K$ and $E_{\pi_0}$ is an independent vector, this coefficient must vanish, $a_{\pi_0}=0$. This yields a contradiction of the initial assumption of the existence of a minimal $a_{\pi_0}\neq 0$. This forces all $a_\pi$ to be 0, and the relations are linearly independent.

	There is one such relation for each permutation not ending in $K$, namely
	\begin{equation}
		K!-(K-1)!
	\end{equation}
	of them. This proves the claim.
\end{proof}

\begin{lemma}\label{lem:independence}
	The $(K-1)!$ functions
	\begin{equation}
		\{\Xi(\sigma K,\cdot):\sigma\in S_{K-1}\}
	\end{equation}
	are linearly independent over $\mathbb C$ on $H_0$. Thus,
	\begin{equation}
		\dim \operatorname{span}\{\Xi(\pi,\cdot):\pi\in S_K\}\geq (K-1)!.
	\end{equation}
\end{lemma}

\begin{proof}
	For $K=2$ the claim is immediate, so assume $K\geq 3$.
	Suppose that there is a linear dependence
	\begin{equation}\label{eq:linear_dep_fixed_last}
		\sum_{\sigma\in S_{K-1}} c_\sigma\,\Xi(\sigma K,\ell)\equiv 0
		\qquad(\ell\in H_0).
	\end{equation}
	Fix $\sigma^*=(\sigma^*_1,\dots,\sigma^*_{K-1})\in S_{K-1}$ and write
	\begin{equation}
		\pi^*=(\sigma^*_1,\dots,\sigma^*_{K-1},K).
	\end{equation}
	For $1\le q\le K-2$, set
	\begin{equation}
		A_q\coloneqq \{\sigma^*_1,\dots,\sigma^*_q\};
	\end{equation}
	the associated prefix flag is the sequence of subset-sum hyperplanes
	\begin{equation}
		H_{A_1},\ H_{A_2},\ \dots,\ H_{A_{K-2}},
	\end{equation}
	where
	\begin{equation}
		H_{A_q}\coloneqq \left\{\ell\in H_0:\sum_{i\in A_q}\ell_i=0\right\}.
	\end{equation}

	The flag walk works as follows: the first residue is along the hyperplane
	\begin{equation}
		H_{A_1}=\{\ell_{\sigma^*_1}=0\}.
	\end{equation}
	A fixed-last-letter term $\Xi(\tau K,\ell)$ can have a pole there only if its first letter is $\sigma^*_1$. Indeed, the suffix alternative in \cref{lem:pole_support} would require a proper prefix of the word $\tau K$ to contain the final letter $K$, which is impossible. Hence the first residue kills every term whose first letter is not $\sigma^*_1$, and for the surviving terms it peels off this first letter.

	Iteratively, since on the partial flag all previous frequency have been set to 0, the next hyperplane $H_{A_q}$ is equivalent to imposing the vanishing of the single frequency $\ell_{\sigma^*_q}$ in the residual partial sum $\sum_{j=1}^q=0$. Note that, by \cref{lem:pole_support}, a pole can again occur by a prefix- or by a suffix-mechanism: the suffix case is impossible for the same reason as before: the complementary set contains the final letter $K$, whereas a proper prefix of a residual word ending in $K$ cannot contain $K$. Thus survival at the $q$-th residue follows the prefix mechanism and forces the next residual letter to be $\sigma^*_q$.

	If after the residues along $H_{A_1},\dots,H_{A_{q-1}}$, a term has survived, by \cref{lem:res_factor_complete} it must agree with $\pi^*$ in the first $q-1$ positions, and the residual denominator is, up to the already peeled non-zero factor, the Floquet denominator associated with the remaining word. Since the original word belongs to the fixed-last-letter family and none of the peeled letters is $K$, this residual word still ends in $K$.
	In other words, for a term to survive the entire flag, it must satisfy
	\begin{equation}
		P_q(\tau K)=\{\sigma^*_1,\dots,\sigma^*_q\}
		\qquad(1\le q\le K-2).
	\end{equation}
	These nested conditions force the first $K-2$ letters of $\tau K$ to agree with those of $\pi^*$. Since both words end in $K$, the remaining non-final letter is also forced. Hence $\tau K=\pi^*$.
	Thus the iterated residue of \cref{eq:linear_dep_fixed_last} isolates the single term $c_{\sigma^*}\Xi(\pi^*,\ell)$ and gives
	\begin{equation}
		c_{\sigma^*}\,R_{\sigma^*}(\ell)\equiv 0,
	\end{equation}
	where $R_{\sigma^*}$ is non-zero. Hence $c_{\sigma^*}=0$. Since $\sigma^*$ was arbitrary, all coefficients vanish, proving linear independence.
\end{proof}

\begin{corollary}\label{cor:xi_rank_shuffle}
	The rational functions $\{\Xi(\pi,\cdot):\pi\in S_K\}$ span a vector space of dimension
	\begin{equation}
		\dim \operatorname{span}\{\Xi(\pi,\cdot):\pi\in S_K\}=(K-1)!;
	\end{equation}
	where, the family
	\begin{equation}
		\{\Xi(\sigma K,\cdot):\sigma\in S_{K-1}\}
	\end{equation}
	is a basis for this span.
	Moreover, every functional relation
	\begin{equation}
		\sum_{\pi\in S_K} c_\pi\,\Xi(\pi,\cdot)=0
	\end{equation}
	is generated by the shuffle relations of \cref{lem:shuffle_xi}.
\end{corollary}

\begin{proof}
	Let
	\begin{equation}
		W_K\coloneqq \operatorname{span}_{\mathbb C}\{E_\pi:\pi\in S_K\}
	\end{equation}
	and let
	\begin{equation}
		T:W_K\longrightarrow \operatorname{span}_{\mathbb C}\{\Xi(\pi,\cdot):\pi\in S_K\},
		\qquad
		E_\pi\longmapsto \Xi(\pi,\cdot).
	\end{equation}
	Then $\ker T$ is exactly the space of functional relations among the $\Xi(\pi,\cdot)$.

	By \cref{lem:shuffle_xi}, every shuffle relation lies in $\ker T$. By \cref{lem:independent_shuffle_relations}, these shuffle relations contain at least $K!-(K-1)!$ linearly independent elements. Hence
	\begin{equation}
		\dim\ker T\geq K!-(K-1)!,
	\end{equation}
	and therefore
	\begin{equation}
		\dim\operatorname{Im}T
		=
		K!-\dim\ker T
		\leq (K-1)!.
	\end{equation}
	On the other hand, \cref{lem:independence} shows that the $(K-1)!$ functions $\{\Xi(\sigma K,\cdot):\sigma\in S_{K-1}\}$ are linearly independent. Thus
	\begin{equation}
		\dim\operatorname{Im}T\geq (K-1)!.
	\end{equation}
	Combining the two inequalities gives
	\begin{equation}
		\dim\operatorname{Im}T=(K-1)! \qquad \text{and} \qquad 	\dim\ker T = K!-(K-1)!.
	\end{equation}
	Hence the fixed-last-letter family is a basis.

	It remains to identify the full kernel. Since  \cref{lem:independent_shuffle_relations} gives exactly $K!-(K-1)!$ independent shuffle relations inside $\ker T$, these relations form a basis of $\ker T$. Therefore all functional relations among the functions $\Xi(\pi,\cdot)$ are generated by shuffle relations.
\end{proof}

\begin{lemma}\label{lem:lie_shuffle}
	Let $F:S_K\to\mathbb C$ and define
	\begin{equation}
		L_F(A_1,\dots,A_K)
		\coloneqq
		\sum_{\pi\in S_K}F(\pi)A_{\pi(1)}\cdots A_{\pi(K)}.
	\end{equation}
	If $L_F(A_1,\dots,A_K)$ is a multilinear Lie element, then, for every decomposition of $[K]$ into two non-empty ordered words $\alpha$ and $\beta$,
	\begin{equation}
		\sum_{\rho\in\operatorname{Sh}(\alpha,\beta)}F(\rho)=0.
	\end{equation}
\end{lemma}

\begin{proof}
	By Ree's shuffle criterion~\cite{Ree1958,reutenauer1993}, the coefficient vector of a multilinear Lie element annihilates every proper shuffle sum. This property immediately applies to the coefficient $F(\pi)$.
\end{proof}

\begin{lemma}\label{lem:filter_compatibility}
	Let
	\begin{equation}
		\mathcal V_K
		\coloneqq
		\operatorname{span}_{\mathbb C}\{\Xi(\pi,\cdot):\pi\in S_K\}.
	\end{equation}
	A target vector $F:S_K\to\mathbb C$ can be realised by a linear functional on $\mathcal V_K$, meaning that there exists a linear map
	\begin{equation}
		\Lambda:\mathcal V_K\to\mathbb C
	\end{equation}
	such that
	\begin{equation}
		\Lambda(\Xi(\pi,\cdot))=F(\pi)
		\qquad\text{for all } \pi\in S_K,
	\end{equation}
	if and only if $F$ annihilates every linear relation among the functions $\Xi(\pi,\cdot)$; that is, whenever
	\begin{equation}
		\sum_{\pi\in S_K} c(\pi)\,\Xi(\pi,\cdot)\equiv 0,
	\end{equation}
	one has
	\begin{equation}
		\sum_{\pi\in S_K} c(\pi)\,F(\pi)=0.
	\end{equation}
\end{lemma}

\begin{proof}
	First suppose that such a linear map $\Lambda$ exists. If
	\begin{equation}
		\sum_{\pi\in S_K} c(\pi)\,\Xi(\pi,\cdot)\equiv 0,
	\end{equation}
	then applying $\Lambda$ gives
	\begin{equation}
		0
		=
		\Lambda\!\left(\sum_{\pi\in S_K} c(\pi)\,\Xi(\pi,\cdot)\right)
		=
		\sum_{\pi\in S_K} c(\pi)\,\Lambda(\Xi(\pi,\cdot))
		=
		\sum_{\pi\in S_K} c(\pi)\,F(\pi).
	\end{equation}
	Hence $F$ annihilates every relation.

	Conversely, suppose that $F$ annihilates every relation. We define $\Lambda$ on $\mathcal V_K$ as follows. For an element of $\mathcal V_K$ written as
	\begin{equation}
		G=\sum_{\pi\in S_K} c(\pi)\,\Xi(\pi,\cdot),
	\end{equation}
	and we write the map $\Lambda$ as
	\begin{equation}
		\Lambda(G) \coloneqq \sum_{\pi\in S_K} c(\pi)\,F(\pi).
	\end{equation}
	We need to check that this does not depend on the chosen representation of $G$. Suppose we can write it in a different form,
	\begin{equation}
		G=\sum_{\pi\in S_K} c'(\pi)\,\Xi(\pi,\cdot),
	\end{equation}
	then
	\begin{equation}
		\sum_{\pi\in S_K}\bigl(c(\pi)-c'(\pi)\bigr)\Xi(\pi,\cdot)\equiv 0.
	\end{equation}
	By assumption, $F$ annihilates this relation, so
	\begin{equation}
		\sum_{\pi\in S_K}\bigl(c(\pi)-c'(\pi)\bigr)F(\pi)=0,
	\end{equation}
	and therefore
	\begin{equation}
		\sum_{\pi\in S_K}c(\pi)F(\pi)
		=
		\sum_{\pi\in S_K}c'(\pi)F(\pi),
	\end{equation}
	which means that $\Lambda$ is well-defined. It is plainly linear, and for each generator $\Xi(\pi,\cdot)$ we obtain
	\begin{equation}
		\Lambda(\Xi(\pi,\cdot))=F(\pi).
	\end{equation}
\end{proof}

We are finally able to state and prove the main result of this section.

\begin{theorem}\label{thm:filtering}
  Let $U\subset H_0^{(K)}$ be the Zariski-open domain on which all functions $\Xi(\pi,\cdot)$, $\pi\in S_K$, are well-defined. Let
  \begin{equation}
    L_F(A_1,\dots,A_K)\coloneqq \sum_{\pi\in S_K}F(\pi)A_{\pi(1)}\cdots A_{\pi(K)}
  \end{equation}
  be any Lie algebra element $L_F$, expanded in terms of permutations of the generators $A_j$ with corresponding coefficients $F(\pi)$.
  Then we can construct $d=(K-1)!$ integer frequency vectors
  \begin{equation}
    \ellvec^{(1)},\dots,\ellvec^{(d)}\in U\cap\mathbb Z^K
  \end{equation}
  where the set $\mathcal L=\{\ellvec^{(1)},\dots,\ellvec^{(d)}\}$ satisfies \cref{def:admissible_frequency_set}, together with coefficients $\alpha_1,\dots,\alpha_d\in\R$, such that
  \begin{equation}\label{eq:lie_filter_target}
    \forall\pi\in S_K:\qquad \sum_{m=1}^{d}\alpha_m\,\Xi(\pi,\ellvec^{(m)})=F(\pi).
  \end{equation}
\end{theorem}

\begin{proof}
	Set $d\coloneqq (K-1)!$. Since $L_F(A_1,\dots,A_K)$ is a multilinear Lie element, \cref{lem:lie_shuffle} shows that $F$ annihilates every shuffle relation:
	\begin{equation}
		\sum_{\rho\in\operatorname{Sh}(\alpha,\beta)}F(\rho)=0
	\end{equation}
	for every decomposition of $[K]$ into two non-empty ordered words $\alpha$ and $\beta$.
	By \cref{cor:xi_rank_shuffle}, the shuffle relations generate all functional relations among the functions $\Xi(\pi,\cdot)$. Hence, applying \cref{lem:filter_compatibility}, $F$ defines a valid linear target on
	\begin{equation}
		\mathcal V_K\coloneqq \operatorname{span}_{\mathbb C}\{\Xi(\pi,\cdot):\pi\in S_K\}.
	\end{equation}

	It remains to realise this valid target by finitely many integer evaluations. By \cref{cor:xi_rank_shuffle}, $\mathcal V_K$ has dimension $d$ with basis $\{\Xi(\sigma K,\cdot):\sigma\in S_{K-1}\}$. Choose
	\begin{equation}
		\ellvec^{(1)},\dots,\ellvec^{(d)}\in U\cap\mathbb Z^K
	\end{equation}
	such that $\mathcal L=\{\ellvec^{(1)},\dots,\ellvec^{(d)}\}$ satisfies \cref{def:admissible_frequency_set} and the evaluation matrix
	\begin{equation}
		M_{\sigma,m}\coloneqq \Xi(\sigma K,\ellvec^{(m)}),
		\qquad
		\sigma\in S_{K-1},
		\qquad
		m=1,\dots,d,
	\end{equation}
	is invertible\footnote{Such a choice exists because the vanishing determinant defines a proper algebraic exceptional set, while failures of \cref{def:admissible_frequency_set} are contained in a finite union of proper linear subspaces of $(H_0)^d$. Hence, by the same Zariski-density argument as in \cref{cor:separated_integer_frequencies} applied to $(H_0)^d$, we may choose integer points outside all these exceptional sets simultaneously.}.

	The point-evaluation functionals at $\ellvec^{(1)},\dots,\ellvec^{(d)}$ then form a basis of $\mathcal V_K^\ast$. Therefore the linear functional $\Lambda$ associated with $F$ by \cref{lem:filter_compatibility} can be written as
	\begin{equation}
		\Lambda(f)=\sum_{m=1}^{d}\alpha_m f(\ellvec^{(m)})
		\qquad
		(f\in\mathcal V_K)
	\end{equation}
	for suitable coefficients $\alpha_1,\dots,\alpha_d$. Since the evaluation matrix and the target values $F(\pi)$ are real, these coefficients can be chosen in $\R$.

	Applying this identity to $f=\Xi(\pi,\cdot)$ gives
	\begin{equation}
		F(\pi)=\Lambda(\Xi(\pi,\cdot))=\sum_{m=1}^{d}\alpha_m\,\Xi(\pi,\ellvec^{(m)})
		\qquad
		(\pi\in S_K).
	\end{equation}
\end{proof}

\begin{remark}[Frequency dilation]\label{rem:frequency_dilation}
	Let $\ellvec\in U\cap\mathbb Z^k$, where $U$ is the domain in
	\cref{def:H0_U}. For every $s>0$,
	\begin{equation}\label{eq:Xi_dilation}
		\Xi(\pi,s\ellvec)=s^{-(k-1)}\Xi(\pi,\ellvec),
		\qquad \pi\in S_k,
	\end{equation}
	and thus if
	\begin{equation}
		\sum_m\alpha_m\Xi(\pi,\ellvec^{(m)})=F(\pi),
	\end{equation}
	then
	\begin{equation}\label{eq:filter_dilation_identity}
		\sum_m\alpha_m s_m^{k-1}
		\Xi(\pi,s_m\ellvec^{(m)})=F(\pi)
	\end{equation}
	for all positive $s_m$.
	Integer dilations preserve integrality, while joint admissibility is ensured by an appropriate choice of the $s_m$.
\end{remark}

\section{Canonical projection on a Lie polynomial}\label{sec:proof:canonical_projection}

We are now in a position to show that suitable choices of driving frequencies allow arbitrary Lie polynomials to be constructed in the leading-order term of the Magnus expansion.
The key mathematical insight underlying this step is that the sum over permutations in the leading-order Magnus term of \cref{thm:leading_order_term} is closely related to the the canonical projection from the free associative algebra onto its free Lie subalgebra.

\begin{definition}[Canonical projection~\cite{reutenauer1993,reutenauer1986,bandiera2017}]
  \label{def:canonical_projection}
  The \emph{canonical projection} $\pi_1$ is a linear map from the free associative algebra $R\langle \cA \rangle$
  to the Lie algebra $\Lie (\cA)$.
  In particular, when the domain is restricted to the Lie algebra, then $\pi_1$ acts as the identity.

  Formally, on a word $A_1A_2\cdots A_k$, $\pi_1$ acts as:
  \begin{equation}
    \pi_1(A_1A_2\dots A_k) =
    \frac 1 k \sum_{\sigma \in S_k} \frac{(-1)^{\des \sigma}}{\binom{k-1}{\des \sigma}}
    A_{\sigma(1)} \cdot A_{\sigma(2)} \cdots A_{\sigma(k)},
  \end{equation}
  and equivalently in terms of right-ordered nested commutators \cite[Theorem 1.8]{bandiera2017},
  \begin{align}
    \pi_1(A_1A_2\dots A_k)
    &= \frac {1}{k^2} \sum_{\sigma \in S_{k}} w(\sigma) \left[A_{\sigma(1)}, A_{\sigma(2)}, \cdots, A_{\sigma(k)} \right] \\
    &= \frac 1 k \sum_{\sigma \in S_{k-1}} w(\sigma) \left[A_{\sigma(1)}, A_{\sigma(2)}, \cdots, A_k\right],
  \end{align}
  where $w(\sigma)$ are the Eulerian weights of \cref{def:eulerian_weights}.
\end{definition}

By choosing the driving frequencies and amplitudes to generate a Lie polynomial, the canonical projection acts as identity thereupon, and we arrive at the cornerstone theorem of this section: any homogeneous multilinear Lie polynomial of any degree $K$ can be generated at leading-order in the Magnus expansion.

\begin{theorem}\label{thm:leading_order_lie_polynomial}
  Let $\LP(A_1,\dots,A_K) \in \Lie (\cA)$ be a homogeneous multilinear Lie polynomial of degree~$K$.
  Then we can select $d=(K-1)!$ integer frequency vectors $\ellvec^{(1)},\dots,\ellvec^{(d)}$ satisfying \cref{def:admissible_frequency_set}, together with real amplitudes $a_j^{(m)}$, to construct a Hamiltonian $H(t)$ of the form in \cref{def:S-driven_blocks}, such that the leading-order non-zero term in the Magnus expansion of $H(t)$ is
  \begin{equation}\label{eq:leading_order_lie_polynomial}
    M_K(H,T)=(-i)^K(-1)^{\lfloor (K-1)/2\rfloor} \frac{2c^KT}{\omega^{K-1}}\, \LP(A_1,\dots,A_K),
  \end{equation}
  where $c\in \R_+$ is an overall factor that can be chosen freely.
\end{theorem}

\begin{proof}
  Choose frequency coefficients $\ellvec^{(1)},\dots,\ellvec^{(d)}$ and coefficients $\alpha_m$ according to \cref{thm:filtering}, so that
  \begin{equation}
    \sum_{m=1}^{d}\alpha_m\,\Xi(\pi,\ellvec^{(m)}) = F(\pi) \qquad (\pi\in S_K).
  \end{equation}
  For any freely chosen $c>0$, choose real amplitudes $a_j^{(m)}$ such that
  \begin{equation}\label{eq:effective_filter_coefficient}
    \alpha_m\coloneqq\frac{1}{c^K}\prod_{j=1}^{K}a_j^{(m)}.
  \end{equation}
  Such a choice is always possible for every $\alpha_m\in\mathbb R$.

  By \cref{thm:leading_order_term,cor:leading_order_Magnus_K}, the leading-order non-zero terms in the Magnus expansion occur at order $K$, and have the form
  \begin{align}
    M_K(H,T)
    &=
      (-i)^K(-1)^{\lfloor (K-1)/2\rfloor}\frac{2c^KT}{K^2\omega^{K-1}}\sum_{\mu\in S_K}[A_{\mu(1)},\dots,A_{\mu(K)}]\sum_{\sigma\in S_K}w(\sigma)\sum_{m=1}^{d}\alpha_m\Xi(\mu\sigma^{-1},\ellvec^{(m)})\\
    &=
      (-i)^K(-1)^{\lfloor (K-1)/2\rfloor}\frac{2c^KT}{K^2\omega^{K-1}}\sum_{\mu\in S_K}[A_{\mu(1)},\dots,A_{\mu(K)}]\sum_{\sigma\in S_K}w(\sigma) c_L(\mu\sigma^{-1}).
  \end{align}

  Now, for each fixed permutation $\sigma$, set
  \begin{equation}
    \pi=\mu\sigma^{-1},
    \qquad
    \mu=\pi\sigma
  \end{equation}
  and relabel the sums to give
  \begin{equation}\label{eq:ME_filtered}
    M_K(H,T) = (-i)^K(-1)^{\lfloor (K-1)/2\rfloor} \frac{2c^KT}{K^2\omega^{K-1}}
               \sum_{\pi\in S_K} c_L(\pi) \sum_{\sigma\in S_K}
                 w(\sigma) [A_{\pi(\sigma(1))},\dots,A_{\pi(\sigma(K))}].
  \end{equation}

  For fixed $\pi$, define
  \begin{equation}
    B_j\coloneqq A_{\pi(j)}
    \qquad
    j=1,\dots,K
  \end{equation}
  so that the commutator form of the canonical projection in \cref{def:canonical_projection} can be rewritten as
  \begin{equation}
    \sum_{\sigma\in S_K}w(\sigma)[A_{\pi(\sigma(1))},\dots,A_{\pi(\sigma(K))}]=\sum_{\sigma\in S_K}w(\sigma)[B_{\sigma(1)},\dots,B_{\sigma(K)}]=K^2\pi_1(A_{\pi(1)}\cdots A_{\pi(K)}).
  \end{equation}
  By linearity of the canonical projection $\pi_1$, \cref{eq:ME_filtered} reduces to
  \begin{align}
    M_K(H,T) &= (-i)^K(-1)^{\lfloor (K-1)/2\rfloor}\frac{2c^KT}{\omega^{K-1}}\;
                \pi_1\Bigl(\sum_{\pi\in S_K} c_L(\pi) A_{\pi(1)}\cdots A_{\pi(K)}\Bigr)\\
             &= (-i)^K(-1)^{\lfloor (K-1)/2\rfloor}\frac{2c^KT}{\omega^{K-1}}\;
                \pi_1\bigl(\LP (A_1,\dots,A_K)\bigr).
  \end{align}
  Since $\LP$ is an element of the Lie algebra, the canonical projection acts as identity on it:
  \begin{equation}
    \pi_1\bigl(\LP(A_1,\dots,A_K)\bigr) = \LP(A_1,\dots,A_K),
  \end{equation}
  and the \lcnamecref{thm:leading_order_lie_polynomial} follows.
\end{proof}

\section{Floquet universality}\label{sec:proof:Floquet_universality}

\subsection{Mixed-degree synthesis and truncation}

Here we extend the construction from homogeneous Lie polynomials to target Hamiltonians that contain Lie-polynomial contributions of different degrees.
Up to this point, we have constructed a Lie polynomial of a fixed degree $K$, or, equivalently, a Hamiltonian whose leading non-zero Magnus contribution occurs at a common order $K$.
We now combine such leading-order contributions of different degrees, to produce a target Hamiltonian of the form:
\begin{equation}\label{eqn:target_hamiltonian}
  \Htarg = \sum_{r=1}^R \imath_r \lambda_r \LP_r,
\end{equation}
where each $\LP_r$ is a homogeneous multilinear Lie polynomial of degree $K_r$ and $\imath_r\in\{1,i\}$ is chosen such that $\Htarg$ is Hermitian.

Recall that for each target Lie polynomial $\LP_r$ of degree $K_r$, the construction of \cref{thm:leading_order_lie_polynomial} uses an admissible frequency set of size $d_r=(K_r-1)!$:
\begin{equation}
	\mathcal L^{(r)}=\left\{\ellvec^{(r,1)},\dots,\ellvec^{(r,d_r)}\right\}.
\end{equation}
Here, $r$ labels the target Lie polynomial, $m$ labels the frequency vectors in that Lie polynomial's construction, and $j$ labels a component of that vector.
Each frequency set $\mathcal L^{(r)}$ is admissible, so by \cref{def:admissible_frequency_set} its only nontrivial zero-sum relations are the designed relations
\begin{equation}
  \sum_{j=1}^{K_r}\ell_j^{(r,m)}=0.
\end{equation}

When combining Lie polynomials of different degrees, however, there is no longer a common value of $K$ governing all the frequency vectors.
In particular, we must handle sums involving frequencies corresponding to different Lie polynomials of different orders, including sums containing up to
\begin{equation}
  K_{\max}\coloneqq\max_{1\leq r\leq R} K_r
\end{equation}
components.
To this end, we impose the following additional separation condition on the frequencies.

\begin{condition}[Admissible frequency families]\label{cond:admissible_frequency_family}
We say that a family of frequency sets
\begin{equation}
  \left\{\mathcal L^{(r)}\right\}_{r=1}^R
\end{equation}
consisting of admissible frequency sets
\begin{equation}
  \mathcal L^{(r)} = \left\{ \ellvec^{(r,m)} \right\}_{m=1}^{(K_r-1)!}, \qquad \ellvec^{(r,m)} = \left( \ell_j^{(r,m)}\right)_{j=1}^{K_r}.
\end{equation}
is \emph{admissible} if, for every $1\leq k\leq K_{\max}$ and every choice of signs $s_1,\dots,s_k\in\{\pm1\}$, the relation
\begin{equation}
  \sum_{q=1}^k s_q\ell_{j_q}^{(r_q,m_q)}=0
\end{equation}
holds iff the terms admit a partition into subsets of the following two types:
\begin{enumerate}[label=(\roman*)]
\item all components of a single frequency vector with a common sign,
  \begin{equation}
    \left\{s\ell_1^{(r,m)},\dots,s\ell_{K_r}^{(r,m)}\right\},
    \qquad s\in\{\pm1\};
  \end{equation}
\item pairs of identical labelled frequency components with opposite signs.
\end{enumerate}
\end{condition}

\noindent A constructive procedure for generating admissible frequency families satisfying this global constraint is given in \cref{sec:spectral_spacing}, based on a uniform recursive spacing factor.
Whilst we do not claim this construction is optimal, it provides a constructive witness for the Zariski argument in \cref{thm:filtering}, establishing the existence of infinitely many such admissible integer frequency configurations.

This new admissibility condition prevents unintended interactions between frequency sets corresponding to different Lie polynomials, as we will show shortly.
Nonetheless, even with this condition, intermediate orders in the Magnus expansion can still give rise to unwanted mixed terms involving different frequency sets.
These mixed terms contribute to the error and are accounted for by a general intermediate-order error bound.

One challenge of combining Lie polynomials of different degrees is that their leading Magnus contributions carry different powers of the fundamental frequency.
We will show that these degree-dependent scalings can be compensated for, and the correct relative weights of the target terms can be generated, via integer dilations of the driving frequencies alone, without scaling the amplitudes.

\subsection{Floquet synthesis with bounded Fourier-amplitudes}

Throughout this subsection, we will assume
\begin{equation}
  2\leq K_r\leq K\coloneqq K_{\max},
\end{equation}
that the $A_j^{(r)}$ are normalised such that $\norm{A_j^{(r)}} \leq 1$, and that the Lie polynomials have real coefficients.
We set $\imath_r=1$ for odd $K_r$ and $\imath_r=i$ for even $K_r$, and keep the base frequency $\omega>0$ fixed.

\begin{remark}[Eliminating degree-one terms]\label{remark:perfect_Lie_algebras}
  The restriction $K_r\geq2$ is without loss of generality whenever $\Lie(\mathcal S)$ is \emph{perfect}, namely,
  \begin{equation}
    [\Lie(\mathcal S),\Lie(\mathcal S)]=\Lie(\mathcal S).
  \end{equation}
  In this case, the linear span of all Lie words over $\mathcal S$ of degree at least two is the full algebra.

  In particular, $\mathfrak{su}(D)$ is perfect, so this restriction plays no role in the key case of Floquet-universal $\cS$-driven Hamiltonians, where $\Lie(\cS) = \mathfrak{su}(D)$.
\end{remark}

The next technical result shows that by discarding terms with amplitudes below a cut-off, the degree-dependent scalings in the leading Magnus terms can be compensated for by suitably dilating the frequency vectors, whilst maintaining a constant bound on all amplitudes.

\begin{definition}[Truncated Hamiltonian]\label{def:truncated_H}
  Let
  \begin{equation}
    \Htarg = \sum_{r=1}^{R}\imath_r\lambda_r\LP_r
  \end{equation}
  be a decomposition of the Hamiltonian $\Htarg$ into homogeneous multilinear Lie polynomials $\LP_r$, where $\lambda_r\in\mathbb{R}$ and $\imath_r\in\{1,i\}$.

  For each $\LP_r$, let $\ellvec^{(r,m)}$ and $\alpha_m^{(r)}$ be the parameters constructed in \cref{thm:filtering}, so that the $d_r=(K_r-1)!$ reference frequency vectors $\ellvec^{(r,m)}$ satisfy the admissibility condition of \cref{lem:integer_dilation_intervals} up to order $K$.
  Let
  \begin{equation}\label{eq:def_w_rm}
    w_{r,m}\coloneqq\lambda_r\alpha_m^{(r)}
  \end{equation}
  and
  \begin{equation}\label{eq:reference_column_polynomial}
    \LP_{r,m}\coloneqq\sum_{\pi\in S_{K_r}}\Xi(\pi,\ellvec^{(r,m)})A_{\pi(1)}^{(r)}\cdots A_{\pi(K_r)}^{(r)}.
  \end{equation}
  so that
  \begin{equation}\label{eq:target_filtering_decomposition}
    \Htarg = \sum_{r,m}\imath_r w_{r,m}\LP_{r,m}.
  \end{equation}

  Then for a cutoff tolerance $\eta>0$, we define the corresponding \emph{truncated Hamiltonian} to be
  \begin{equation}\label{eq:truncated_target}
    \hatHtarg \coloneqq \sum_{(r,m)\in\mathcal I} \imath_r w_{r,m}\LP_{r,m},
  \end{equation}
  where
  \begin{equation}\label{eq:filter_cutoff_threshold}
    \mathcal I\coloneqq \left\{(r,m) : \abs{w_{r,m}} \geq \vartheta\right\},
    \qquad
    \vartheta\coloneqq\frac{\eta}{K! \sum_r d_r}.
  \end{equation}
  Note that, since $\|\LP_{r,m}\| \leq K!$, we have the approximation
	\begin{equation}\label{eq:filter_cutoff_error}
	  \big\|\Htarg - \hatHtarg\big\| \leq \eta.
	\end{equation}
\end{definition}

The truncation isolates a finite set of synthesised contributions whose omission changes the target Hamiltonian by at most $\eta$. This allows us to transfer the weights of the remaining contributions to the frequency dilations while keeping all physical amplitudes within a constant ratio. That is,

\begin{lemma}\label{lem:bounded_amplitude_synthesis}
  For any $\eta > 0$, let $\hatHtarg$ be the truncated Hamiltonian corresponding to $\Htarg$ with cut-off $\eta$, and let $\mathcal I$ be the set of vectors surviving the cut-off in \cref{eq:filter_cutoff_threshold}. Define
   \begin{equation}
 	V\coloneqq|\mathcal I|,
 	\qquad
 	w_\ast\coloneqq\max_{(r,m)\in\mathcal I}|w_{r,m}|.
   \end{equation}
   together with the spacing constant
   \begin{equation}\label{eq:integer_dilation_threshold}
 	L \coloneqq 2+K^2(2K)^K V^{K-1}.
   \end{equation}

  Then for any $0 < y \leq 1$, we can construct Hamiltonians $H^{(r)}(t;y)$ of the form in \cref{def:S-driven_blocks} with corresponding amplitudes $\widetilde a_j^{(r,m)}$ and frequency vectors~$\widetilde\ellvec^{(r,m)}$, such that the Hamiltonians $H^{(r)}(t)$ have individual leading-order Magnus contributions $\cM_{K_r}^{(r)}(T)$ satisfying
  \begin{equation}\label{eq:bounded_amplitude_leading}
    \sum_r \cM_{K_r}^{(r)}(T) = -i\delta \hatHtarg,
  \end{equation}
  with the parameter
  \begin{equation}
	\delta=\frac{4\pi y^K}{w_\ast L^{K-1}}.
  \end{equation}
  Moreover, for $c_0\coloneqq\omega y$, every amplitude $\widetilde a_j^{(r,m)}$ satisfies
  \begin{equation}\label{eq:bounded_amplitude_ratio}
    c_0 \leq |\widetilde a_j^{(r,m)}| < 2c_0,
    \qquad
    \frac{\max_{(r,m)\in\mathcal I,j}|\widetilde a_j^{(r,m)}|}{\min_{(r,m)\in\mathcal I,j}|\widetilde a_j^{(r,m)}|}<2.
  \end{equation}
\end{lemma}

\begin{proof}
  Fix
  \begin{equation}\label{eq:bounded_amplitude_signal}
    s_{r,m}^{\star} = \left(\frac{4\pi y^{K_r}}{\delta|w_{r,m}|}\right)^{1/(K_r-1)}.
  \end{equation}
  As proven in later \cref{lem:integer_dilation_intervals}, if the dilation parameters satisfy $s_v^\star\geq L$, then one can choose integers within the interval
  \begin{equation}\label{eq:bounded_amplitude_integer_dilation}
  	s_v\in[s_v^\star,2s_v^\star]\cap\mathbb N
  \end{equation}
  and set
  \begin{equation}\label{eq:bounded_amplitude_dilations}
    \widetilde\ell_j^{(r,m)} \coloneqq s_{r,m}\ell_j^{(r,m)},  
  \end{equation}
  so that the frequencies are jointly admissible as per \cref{cond:admissible_frequency_family}.

  Finally, define
  \begin{equation}
  	h_{r,m}\coloneqq \left(\frac{s_{r,m}}{s_{r,m}^{\star}}\right)^{(K_r-1)/K_r}.
  \end{equation}
  and choose the amplitudes as
  \begin{equation}\label{eq:bounded_amplitude_signs}
    \begin{split}
      \widetilde a_1^{(r,m)} &\coloneqq \sigma_{r,m} c_0 h_{r,m}, \quad \sigma_{r,m}\coloneqq(-1)^{K_r+1}\operatorname{sgn}(w_{r,m}),\\
      \widetilde a_j^{(r,m)} &\coloneqq c_0h_{r,m}, \quad 2\leq j\leq K_r
    \end{split}
  \end{equation}

  From \cref{eq:bounded_amplitude_signal},
  \begin{equation}
    s_{r,m}^{\star} =
    \left(\frac{w_\ast}{|w_{r,m}|}\right)^{1/(K_r-1)} L^{(K-1)/(K_r-1)} y^{-(K-K_r)/(K_r-1)}
    \geq L,
  \end{equation}
  so the choice in \cref{eq:bounded_amplitude_integer_dilation,eq:bounded_amplitude_dilations} is admissible.
  Recall from \cref{thm:leading_order_lie_polynomial} that, after dilating the frequencies by $s_{r,m}$ and choosing amplitudes of magnitude $\omega y h_{r,m}$, the magnitude of the leading-order coefficient associated with the vector $(r,m)$ is
  \begin{equation}
  	\frac{4\pi(yh_{r,m})^{K_r}}{s_{r,m}^{K_r-1}}.
  \end{equation}
  The factor $h_{r,m}$ was chosen precisely to compensate for replacing the ideal dilation $s_{r,m}^\star$ by the nearby integer $s_{r,m}$. Indeed,
  \begin{equation}\label{eq:integer_calibration_identity}
  	\frac{4\pi(yh_{r,m})^{K_r}}{s_{r,m}^{K_r-1}}
  	=
  	\frac{4\pi y^{K_r}}{(s_{r,m}^{\star})^{K_r-1}}
  	=
  	\delta|w_{r,m}|.
  \end{equation}
  Together with the sign choice in \cref{eq:bounded_amplitude_signs}, this gives the contribution $-i\delta\imath_r w_{r,m}\LP_{r,m}$ vector $(r,m)\in\mathcal I$.
  Summing over $(r,m)\in\mathcal I$ yields \cref{eq:bounded_amplitude_leading}.

  Finally,
  \begin{equation} 1\leq s_{r,m}/s_{r,m}^{\star}\leq2. \end{equation}
  and therefore
  \begin{equation} 1\leq h_{r,m}<2,  \end{equation}
  which implies \cref{eq:bounded_amplitude_ratio} since $|\widetilde a_j^{(r,m)}|=c_0h_{r,m}$.
\end{proof}

The preceding lemma transfers the relative scale of the coefficients surviving the cutoff to the integer frequency dilations, while keeping all non-zero physical amplitudes within a ratio of~2. It remains to control the higher Magnus orders: for this, we need a bound showing that whenever all components of a dilated frequency vector occur in a higher-order time-ordered integral, the contribution retains the factor $s^{-(k-1)}$.

\begin{lemma}\label{lem:resonant_frequency_suppression}
  Fix $(r,m)\in\mathcal I$, set $k=K_r$, and let $\widetilde\ell^{(r,m)}=s_{r,m}\ell^{(r,m)}$ be as in \cref{eq:bounded_amplitude_dilations}. For $J\geq k$, consider the time-ordered integral
  \begin{equation}\label{eq:resonant_integral}
    I_J
    =
    \int_0^Tdt_1\cdots\int_0^{t_{J-1}}dt_J\,
    \exp\left(i\omega\sum_{\nu=1}^{J}a_\nu t_\nu\right),
    \qquad
    T=\frac{2\pi}{\omega},
  \end{equation}
  where $a_1,\dots,a_J\in\mathbb Z$. Suppose that $k$ of the entries
  $a_\nu$ are, in arbitrary order,
  \begin{equation}
    \epsilon_1\widetilde\ell_1^{(r,m)},\dots,
    \epsilon_k\widetilde\ell_k^{(r,m)},
    \qquad
    \epsilon_j\in\{\pm1\},
  \end{equation}
  while the remaining $J-k$ entries are arbitrary. Then
  \begin{equation}\label{eq:resonant_integral_bound}
    |I_J| \leq \frac{e^{\pi/2}(J+1)4^J}{\omega^J s_{r,m}^{k-1}}.
  \end{equation}
\end{lemma}

\begin{proof}
  Write $s=s_{r,m}$ and $\ell_j=\ell_j^{(r,m)}$ for brevity. Rescale the integration variables as $\theta_\nu=\omega t_\nu$, and define
  \begin{equation}
    S_0\coloneqq0,
    \qquad
    S_q\coloneqq\sum_{\nu=1}^{q}a_\nu,
    \qquad q=1,\dots,J.
  \end{equation}
  Then
  \begin{equation}
    \omega^J I_J
    =
    \int_0^{2\pi}d\theta_1\cdots \int_0^{\theta_{J-1}}d\theta_J\, \exp\left(i\sum_{\nu=1}^{J}a_\nu\theta_\nu\right).
  \end{equation}
  Using the contour representation of the time-ordered integral,
  \begin{equation}\label{eq:ordered_integral_contour}
    \omega^J I_J
    =
    \sum_{u\in\{S_0,\dots,S_J\}} \frac{1}{2\pi i} \oint_{|z-iu|=1/4} 
    \frac{e^{2\pi z}}{\prod_{q=0}^{J}(z-iS_q)}\,dz,
  \end{equation}
  where the sum is over the distinct values among $S_0,\dots,S_J$.

  Fix one such value $u$ and regard $S_0,\dots,S_J$ as the vertices of a path whose edges are the increments $a_1,\dots,a_J$. Remove the $J-k$ edges not belonging to the selected dilated vector. The remaining path has at most $J-k+1$ connected components.

  Within any component containing fewer than all $k$ selected edges, the difference between two distinct vertices is
  \begin{equation}
    s\sum_{j\in Q}\epsilon_j\ell_j,
    \qquad
    \varnothing\neq Q\subsetneq\{1,\dots,k\}.
  \end{equation}
  By admissibility of the reference vector, this is a non-zero integer multiple of $s$, and hence has modulus at least $s$. If one component contains all $k$ selected edges, only the difference between its two endpoints can involve all components; every other non-trivial difference again has modulus at least $s$.

  Consequently, at most $J-k+2$ of the $J+1$ vertices can lie within distance strictly smaller than $s/2$ from $u$. Hence at least $k-1$ indices satisfy
  \begin{equation}
    |S_q-u|\geq\frac{s}{2}.
  \end{equation}
  On $|z-iu|=1/4$, the corresponding denominator factors obey
  \begin{equation}
    |z-iS_q| \geq \frac{s}{2}-\frac{1}{4} \geq \frac{s}{4}.
  \end{equation}
  Every remaining denominator has modulus at least $1/4$, since the $S_q$ are integers, while $|e^{2\pi z}|\leq e^{\pi/2}$. Thus each contour contributes at most
  \begin{equation}
    e^{\pi/2}4^J s^{-(k-1)}.
  \end{equation}
  There are at most $J+1$ distinct contours, and therefore
  \begin{equation}
    |I_J| \leq \frac{e^{\pi/2}(J+1)4^J}{\omega^J s^{k-1}}
    =
    \frac{e^{\pi/2}(J+1)4^J}{\omega^J s_{r,m}^{k-1}},
  \end{equation}
  which proves \cref{eq:resonant_integral_bound}.
\end{proof}

Next, we will control both the tail bound of the Magnus expansion and the intermediate terms mixing letters from different Lie polynomials. For the first task, we will need the truncation bound presented in \cite{ACO}, namely

\begin{theorem}[Error bound for truncated Magnus series, Theorem~2 of \cite{ACO}]\label{thm:ACO}
	Consider a general time-dependent Hamiltonian $H$, and assume that 
	\begin{equation}\label{eq:A_t_definition}
		A(t) \coloneqq \int_0^t \norm{H(s)}_\mathrm{op} \, d s <\xi\approx 1.086869.
	\end{equation}
	Denote the truncation at term $N$ by $\mathcal{M}^{(N)}(t) = \sum_{n=1}^N M_n(t)$.
	Then the truncation error is bounded by
	\begin{equation}\label{eq:main_truncation_bound_closed}
		\norm{\mathcal{M}(t) - \mathcal{M}^{(N)}(t)}_\mathrm{op}
		\leq
		\frac{K_{\mc M}}{(N+1)^{3/2}} \frac{\left(A(t)/\xi\right)^{N+1}}{1-A(t)/\xi} ,
	\end{equation}
	where $K_{\mc M}$ is a global constant.
\end{theorem}
\noindent In particular, $A(t)/\xi<1$ is a sufficient criterion for convergence of the Magnus series.

\medskip
We can now state:

\begin{theorem}\label{thm:bounded_amplitude_per_period}
  Let
  \begin{equation}
    \Htarg = \sum_{r=1}^R\imath_r \lambda_r\LP_r
  \end{equation}
  be a decomposition of a Hamiltonian $\Htarg$ into homogeneous multilinear Lie polynomials.
  Let $w_{r,m}$ be the coefficients in \cref{def:truncated_H} of the truncated Hamiltonian, and define $w_\ast\coloneqq\max_{r,m}\abs{w_{r,m}}$.

  For any multiplicative error tolerance $\epsilon>0$ and any $0 < y \leq y_\ast(\epsilon)$ with a threshold $y_\ast(\epsilon)$ for which we can give an explicit expression, we can construct a Hamiltonian $H(t)$ of the form of \cref{def:S-driven_blocks} such that the Magnus series for $H(t)$ converges and satisfies the multiplicative error bound
  \begin{equation}\label{eq:bounded_amplitude_per_period_error}
    \big\|i\mc M(T)-\delta \Htarg\big\| \leq \epsilon \, \delta\|\Htarg\|,
  \end{equation}
  where $\delta = \frac{4\pi y^K}{w_\ast L^{K-1}}$.
\end{theorem}

\begin{proof} 
	We decompose the Magnus series as
	\begin{equation}\label{eq:bounded_amplitude_remainders}
	\mc M(T) = -i\delta\hatHtarg + \mc R_{\mathrm{int}}(T) + \mc R_{\mathrm{tail}}(T),
	\end{equation}
	where $\mc R_{\mathrm{int}}(T)$ contains all the unwanted contributions of order at most $K$, and $\mc R_{\mathrm{tail}}(T)$ contains all terms of orders greater than $K$.
	
	Choose the cutoff
	\begin{equation}\label{eq:bounded_amplitude_cutoff_choice}
	\eta \coloneqq \frac{\epsilon}{2}\|\Htarg\|,
	\end{equation}
	and define $P\coloneqq\sum_r K_r d_r$ and
	\begin{align}
	 	\mathcal E_{\mathrm{int}}
	 	&\coloneqq
	 	\frac{w_\ast e^{\pi/2}}{8\pi} \sum_{J=2}^{K} \frac{(J+1)J!}{J^2}(32P)^J,\label{eq:E_int}\\
	 	\mathcal E_{\mathrm{tail}} 
	 	&\coloneqq
	 	\frac{w_\ast L^{K-1}}{4\pi} \frac{2K_M}{(K+1)^{3/2}} \left(\frac{8\pi P}{\xi}\right)^{K+1}. \label{eq:E_tail}
	\end{align}
	
	\medskip
	We first address convergence and the tail bound using \cref{thm:ACO}. From \cref{eq:bounded_amplitude_ratio}, every elementary amplitude has magnitude smaller than $2c_0$. Since at most $P$ elementary components are present, the integrated norm in \cref{eq:A_t_definition} satisfies
	\begin{equation}
	 	A(T) = \int_0^T \|\widetilde H(s;y)\|\,d s \leq 4Pc_0T = 8\pi Py.
	\end{equation}
	For $y\leq1/(32\pi P)$, we have $A(T)/\xi\leq1/(4\xi)<1/2$, so the Magnus series converges. Applying \cref{thm:ACO} with $N=K$ gives
	\begin{align}
		\|\mc R_{\mathrm{tail}}(T)\|
		&\leq
		\frac{K_{\mc M}}{(K+1)^{3/2}} \frac{(A(T)/\xi)^{K+1}}{1-A(T)/\xi} \label{eq:bounded_amplitude_tail_error_1} \\
		&\leq
		\frac{2 K_{\mc M}}{(K+1)^{3/2}}\left(\frac{8\pi Py}{\xi}\right)^{K+1}
		= \delta y\,\mathcal E_{\mathrm{tail}}.\label{eq:bounded_amplitude_tail_error_2}
	\end{align}
	where the last equality follows from $\delta=4\pi y^K/(w_\ast L^{K-1})$ and the definition of~$\mathcal E_{\mathrm{tail}}$ in \cref{eq:E_tail}.

	\medskip
	We now consider intermediate orders, namely $J\leq K$, and consider a time-ordered integral that omits at least one component of every frequency vector. Since none of the designed complete resonances can occur, joint admissibility implies that any vanishing frequency sum arising in the recursive integration consists only of opposite-sign pairs. If no such vanishing sum occurs, the integral vanishes over one period by \cref{lem:vanishing_nested_integrals}; otherwise, its resonant contributions cancel with their mirror partners as in the proof of \cref{thm:vanishing_lower_orders}. In conclusion, only integrals containing all $K_r$ components of at least one complete vector $\ellvec^{(r,m)}$ surviving the cut-off can contribute for $J\leq K$.
	
	If $J=K_r$, the corresponding contribution belongs to the isolated leading term already subtracted in \cref{eq:bounded_amplitude_remainders}, so every integral contributing to $\mc R_{\mathrm{int}}(T)$ must contain a complete vector with $J\geq K_r+1$.
	For each such contribution, choose one complete vector labelled by $(r,m)$, so that the corresponding commutator term contains an occurrence of each of its $K_r$ components, with arbitrary signs and positions. Therefore, by \cref{lem:resonant_frequency_suppression}, the associated time-ordered integral $I_J$ satisfies
	\begin{equation}
		|I_J| \leq \frac{e^{\pi/2}(J+1)4^J}{\omega^J s_{r,m}^{K_r-1}}.
	\end{equation}
	
	Multiplying by the bound $(2\omega y)^J$ on the product of the $J$ amplitude magnitudes and using \cref{eq:integer_calibration_identity} gives
	\begin{equation}\label{eq:bounded_amplitude_resonant_term}
		|I_J|(2\omega y)^J
		\leq e^{\pi/2}(J+1)8^J\frac{y^J}{s_{r,m}^{K_r-1}}
		\leq \frac{\delta w_\ast y e^{\pi/2}}{4\pi}(J+1)8^J,
	\end{equation}
	where we used
	\begin{equation}\label{eq:bounded_amplitude_extra_power}
		\frac{y^J}{s_{r,m}^{K_r-1}}
		= \frac{\delta|w_{r,m}|}{4\pi h_{r,m}^{K_r}}y^{J-K_r}
		\leq \frac{\delta w_\ast y}{4\pi},
	\end{equation}
	since $h_{r,m}\geq1$, $J\geq K_r+1$, and $y\leq1$.
	
	There are at most $P^J$ ordered component and $2^J$ sign choices, and $J!$ permutations. The Eulerian commutator coefficient contributes a factor $1/J^2$, while the commutator norm contributes $2^{J-1}$. Combining the combinatorial factors with \cref{eq:bounded_amplitude_resonant_term} gives
	\begin{align}
		\|\mc R_{\mathrm{int},J}(T)\|
		&\leq \frac{J!}{J^2}P^J2^J2^{J-1} \frac{\delta w_\ast y e^{\pi/2}}{4\pi}(J+1)8^J\\
		&= \frac{\delta w_\ast y e^{\pi/2}}{8\pi} \frac{(J+1)J!}{J^2}(32P)^J.
	\end{align}
	Summing over $J=2,\dots,K$ and using \cref{eq:E_int} therefore gives
	\begin{equation}\label{eq:bounded_amplitude_intermediate_error}
		\|\mc R_{\mathrm{int}}(T)\|
		\leq \frac{\delta w_\ast y e^{\pi/2}}{8\pi} \sum_{J=2}^{K}\frac{(J+1)J!}{J^2}(32P)^J
		= \delta y \,\mathcal E_{\mathrm{int}}.
	\end{equation}
	
	\medskip

	Bringing together the two bounds, we set $\mathcal E\coloneqq\mathcal E_{\mathrm{int}}+\mathcal E_{\mathrm{tail}}$ and choose the threshold
	\begin{equation}\label{eq:bounded_amplitude_y_threshold} 
		y_\ast(\epsilon) \coloneqq \min\left\{\frac{1}{32\pi P},\frac{\epsilon\|\Htarg\|}{2\mathcal E}\right\}.
	\end{equation}
	This ensures that for every $y\leq y_\ast(\epsilon)$ the Magnus series converges and also, combining \cref{eq:bounded_amplitude_tail_error_2,eq:bounded_amplitude_intermediate_error}, that
	\begin{equation}\label{eq:bounded_amplitude_dynamic_error}
		\big\| \mc M(T) + i\delta\hatHtarg \big\| \leq \delta y\mathcal E
		\leq \frac{\epsilon}{2} \delta\|\Htarg\|.
	\end{equation}
	Finally, with \cref{eq:filter_cutoff_error,eq:bounded_amplitude_cutoff_choice},
	\begin{align}
		\big\|\mc M(T) + i\delta\Htarg\big\|
		&\leq \big\|\mc M(T) + i\delta\hatHtarg\big\| + \delta\big\|\Htarg - \hatHtarg\big\| \\
		&\leq \frac{\epsilon}{2}\delta\|\Htarg\| + \frac{\epsilon}{2}\delta\|\Htarg\| \\
		&= \epsilon\delta\|\Htarg\|.
	\end{align}
\end{proof}

\begin{remark}
  One can make $c_0$ in \cref{eq:bounded_amplitude_ratio} arbitrarily small by choosing a sufficiently large integer $N$, since
  \begin{equation}
    c_0=\omega y = \omega\left(\frac{w_\ast L^{K-1}}{4\pi N}\right)^{1/K}.
  \end{equation}
  The reduced Hamiltonian can then be amplified over the runtime of multiple periods $NT$.
\end{remark}

We finally arrive at our main Floquet Hamiltonian simulation result.

\begin{theorem}[Floquet universality]\label{thm:bounded_amplitude_universal}
  For any $\Htarg\in\Lie(\mathcal S)$ where $\Lie(\cS)$ is perfect, and any $\epsilon>0$, the construction of \cref{thm:bounded_amplitude_per_period} yields an $\cS$-driven Hamiltonian of the form in \cref{def:S-driven_blocks} that satisfies
  \begin{equation}\label{eq:universal_simulation}
    \norm{i \mc M(N T)-\Htarg} \leq \epsilon\norm{\Htarg},
  \end{equation}
  where $\mc M(NT)$ is the Magnus expansion evaluated at a multiple of $N$ periods.
  Moreover, the ratios of the amplitudes in \cref{def:S-driven_blocks} satisfy $1 \leq |a_j^{(m)}|/|a_k^{(n)}| < 2$ for any two $|a_j^{(m)}|\geq |a_k^{(n)}|$
 
  \medskip 
  In particular, any $\cS$-driven Hamiltonian for which $\Lie(\mathcal S)=\mathfrak{su}(D)$ for every system size $D$ is a universal Floquet simulator.
\end{theorem}

\begin{proof}
  Since $\Lie(\mathcal S)$ is a perfect Lie algebra, $\Htarg$ admits a finite Lie-polynomial decomposition with all degrees at least two.
  Fix such a decomposition and let $K$ denote its maximal degree.
  Choose an integer
  \begin{equation}
    N\geq\frac{w_\ast L^{K-1}}{4\pi y_\ast(\epsilon)^K}
  \end{equation}
  and set
  \begin{equation}
    y\coloneqq\left(\frac{w_\ast L^{K-1}}{4\pi N}\right)^{1/K},
  \end{equation}
  so that $y\leq y_\ast(\epsilon)$ from \cref{eq:bounded_amplitude_y_threshold} and
  \begin{equation}\label{eq:N_delta_1}
    \delta=\frac{4\pi y^K}{w_\ast L^{K-1}}=\frac1N.
  \end{equation}
  Then by \cref{thm:bounded_amplitude_per_period}
  \begin{equation}
    \big\|i\mc M(T)-\delta \Htarg\big\|
    \leq\epsilon\delta\|\Htarg\|.
  \end{equation}

  Since the Hamiltonian is $T$-periodic, setting $\tau=NT$ gives
  \begin{equation}
    U(\tau)=U(T)^N=\exp\!\big(N\mc M(T)\big),
  \end{equation}
  and $\mc M(\tau)=N\mc M(T)$ is an anti-Hermitian exponent for the evolution over the physical time $\tau$.
  From \cref{eq:N_delta_1} we have $N\delta=1$ and obtain
  \begin{equation}
    \big\|i\mc M(\tau)-\Htarg\big\| = N\big\|i\mc M(T)-\delta \Htarg\big\|
    \leq \epsilon\|\Htarg\|.
  \end{equation}

  Finally, if $\Lie(\mathcal S)=\mathfrak{su}(D)$, by \cref{remark:perfect_Lie_algebras} every target admits a decomposition with $K_r\geq2$.
  The amplitude-ratio bound is inherited from \cref{lem:bounded_amplitude_synthesis}.
\end{proof}

\noindent As an immediate consequence, there exists an anti-Hermitian operator $\Omega$ such that
\begin{equation}
  U(\tau)=\exp\!\big(-i\Htarg+\Omega\big),
  \qquad \text{with}\quad
  \|\Omega\|\leq\epsilon\|\Htarg\|.
\end{equation}

We conclude the section with a scaling for the frequency coefficients in terms of the parameters $K$, $R$ and $\epsilon$.

\begin{lemma}\label{lem:dilated_frequency_scaling}
  Let the target Hamiltonian,
  \begin{equation}
    \Htarg = \sum_{r=1}^R\imath_r \lambda_r\LP_r ,
  \end{equation}
  be a decomposition of a Hamiltonian $\Htarg$ into homogeneous multilinear Lie polynomials, with $2\leq K_r\leq K$, $0<\epsilon\leq1$, and $|\lambda_r|\leq\lambda_+$, where $\omega>0$ and $\lambda_+>0$ are fixed.
  Let $w_{r,m}$ be the coefficients in \cref{def:truncated_H} of the truncated Hamiltonian and define $w_\ast\coloneqq\max_{r,m}\abs{w_{r,m}}$.
  Define
  \begin{equation}
    f_\ast\coloneqq\max_r\max_{\pi\in S_{K_r}}|F_r(\pi)|.
  \end{equation}

  Then the reference frequencies and the integer dilations in \cref{thm:bounded_amplitude_per_period} can be chosen such that, for some constant $C>0$ depending only on $\omega$ and $\lambda_+$,
  \begin{equation}\label{eq:dilated_frequency_scaling}
    \max_{(r,m)\in\mathcal I,\,1\leq j\leq K_r}|\widetilde\ell_j^{(r,m)}|
    \leq e^{CK^5}R^{K^3} \left(\frac{1+f_\ast}{\epsilon\|\Htarg\|}\right)^{K-1}.
  \end{equation}
\end{lemma}

In particular, note that the scaling of \cref{eq:dilated_frequency_scaling} in terms of $K$, $R$, and $\epsilon$ is
\begin{equation}
  \max_{(r,m)\in\mathcal I,\,1\leq j\leq K_r}|\widetilde\ell_j^{(r,m)}|
  = \mc O\!\left(e^{\mathrm{poly}(K)}R^{K^3}\epsilon^{-K}\right).
\end{equation}
Thus, for $K$ constant, the frequencies are polynomial in both the number of target Lie-polynomial terms and the inverse simulation precision.

\begin{proof}[of \cref{lem:dilated_frequency_scaling}]
  Set $B=4K$.
  For each $r$, write $k=K_r$ and choose the reference vectors indexed by $\tau\in S_{k-1}$ as
  \begin{equation}\label{eq:scaling_reference_vectors}
    \ell_{\tau(j)}^{(r,\tau)} = B^{j-1}\quad(1\leq j\leq k-1),
    \qquad \ell_k^{(r,\tau)}=-\sum_{j=0}^{k-2}B^j.
  \end{equation}
  In increasing-power order, the net multiplicities of any vanishing signed sum of at most $K$ components satisfy
  \begin{equation}
    \sum_{j=1}^{k-1}(n_j-n_k)B^{j-1}=0,
    \qquad |n_j-n_k|\leq K<B.
  \end{equation}
  Successive reduction modulo $B$ gives $n_1=\cdots=n_k$.
  Hence the relation consists of complete vectors with a common sign and opposite-sign pairs of the same labelled component, as required by \cref{lem:integer_dilation_intervals}.
  Also,
  \begin{equation}\label{eq:scaling_reference_size}
    1\leq|\ell_j^{(r,\tau)}|\leq2B^{k-2}.
  \end{equation}

  Consider the matrix
  \begin{equation}\label{eq:scaling_filter_matrix}
    M^{(k)}_{\sigma,\tau}\coloneqq\Xi(\sigma k,\ellvec^{(r,\tau)}),
    \qquad \sigma,\tau\in S_{k-1},
  \end{equation}
  and set $G_k=B^{-(k-1)(k-2)/2}$. For positive $x_1,\ldots,x_n$,
  \begin{equation}\label{eq:scaling_positive_sum_identity}
    \sum_{\sigma\in S_n}\prod_{j=1}^{n}\frac{1}{x_{\sigma(1)}+\cdots+x_{\sigma(j)}}
    =
    \prod_{j=1}^{n}\frac{1}{x_j}.
  \end{equation}
  This follows by induction, grouping permutations by their last entry. For fixed $\sigma$, varying $\tau$ permutes the entries
  $1,B,\ldots,B^{k-2}$, so every row of $M^{(k)}$ has sum $G_k$.
  All entries are positive, and
  \begin{equation}\label{eq:scaling_diagonal_bound}
    \frac{M^{(k)}_{\sigma,\sigma}}{G_k}
    =
    \prod_{j=0}^{k-2}\frac{B^j}{1+B+\cdots+B^j}
    \geq
    (1-B^{-1})^{k-2}>\frac34.
  \end{equation}
  Thus $\|I-M^{(k)}/G_k\|_\infty<1/2$ in the induced infinity norm. The Neumann series and
  \cref{cor:xi_rank_shuffle,lem:filter_compatibility} give a full synthesis with
  \begin{equation}\label{eq:scaling_filter_coefficients}
    \|(M^{(k)})^{-1}\|_\infty\leq2G_k^{-1},
    \qquad \max_m|\alpha_m^{(r)}|\leq2f_\ast B^{(k-1)(k-2)/2},
  \end{equation}
  where joint admissibility is imposed by choosing integer dilations.

  Perform the cutoff with \cref{eq:bounded_amplitude_cutoff_choice}.
  Then \cref{eq:def_w_rm,eq:filter_cutoff_threshold,eq:scaling_filter_coefficients} give
  \begin{align}
    w_\ast &\leq 2\lambda_+f_\ast B^{(K-1)(K-2)/2},\label{eq:scaling_weight_upper}\\ 
   |w_{r,m}|
   &\geq \vartheta = \frac{\epsilon\|\Htarg\|}{2K!\sum_r d_r},
   \qquad
   (r,m)\in\mathcal I, \label{eq:scaling_weight_lower}
  \end{align}
  and the target decomposition (with $\|A_j^{(r)}\|\leq1$) imply
  \begin{equation}\label{eq:scaling_target_norm}
    \|\Htarg\| \leq \sum_r|\lambda_r|\sum_{\pi\in S_{K_r}}|F_r(\pi)| \leq \lambda_+RK!f_\ast.
  \end{equation}

  We denote $d=K^2-K+2$.
  Using $V\leq R(K-1)!$, $P\leq RK!$, and \cref{eq:scaling_weight_upper} in \cref{eq:integer_dilation_threshold,thm:bounded_amplitude_per_period}, we obtain
  \begin{align}
    L &\leq e^{\mc O(K^2\log K)}R^{K-1},\label{eq:scaling_dilation_threshold}\\ \mathcal E
    &\leq e^{\mc O(K^3\log K)}f_\ast R^d.\label{eq:scaling_error_constant}
  \end{align}
  Note that here and below, the implicit constants depend at most on $\omega$ and $\lambda_+$,
  The exponent $d$ comes from the tail factor $L^{K-1}P^{K+1}$, and the intermediate sum contributes at most $R^K$.
  Choose
  \begin{equation}\label{eq:scaling_y0_choice}
    y_0 \coloneqq \min\{y_\ast(\epsilon),(2\omega)^{-1}\}.
  \end{equation}
  By \cref{eq:bounded_amplitude_y_threshold,eq:scaling_target_norm,eq:scaling_error_constant},
  \begin{equation}\label{eq:scaling_y0_bound}
    y_0^{-1} \leq e^{\mc O(K^3\log K)}R^d\frac{1+f_\ast}{\epsilon\|\Htarg\|}.
  \end{equation}
  We now set
  \begin{equation}\label{eq:scaling_runtime_choice}
    N\coloneqq\left\lceil\frac{w_\ast L^{K-1}}{4\pi y_0^K}\right\rceil,
    \qquad y\coloneqq\left(\frac{w_\ast L^{K-1}}{4\pi N}\right)^{1/K},
  \end{equation}
  then $y\leq y_0\leq y_\ast(\epsilon)$, $c_0\leq1/2$, and $N\delta=1$.
  The ceiling and \cref{eq:scaling_weight_lower} give the bound
  \begin{align}\label{eq:scaling_y_bound}
    y^{-1} &\leq y_0^{-1}+\left(\frac{4\pi}{w_\ast L^{K-1}}\right)^{1/K} \\
    &\leq y_0^{-1}+\left(\frac{8\pi K!\sum_r d_r}{\epsilon\|\Htarg\|}\right)^{1/K} \\
    &\leq e^{\mc O(K^3\log K)}R^d\frac{1+f_\ast}{\epsilon\|\Htarg\|},
  \end{align}
  where the last inequality follows from \cref{eq:scaling_target_norm,eq:scaling_y0_bound} and $f_\ast^{1-1/K}\leq1+f_\ast$.

  For each vector of degree $k=K_r$, \cref{eq:bounded_amplitude_signal} yields
  \begin{equation}\label{eq:scaling_ideal_dilation}
    s_{r,m}^\star = \left(\frac{w_\ast}{|w_{r,m}|}\right)^{1/(k-1)} L^{(K-1)/(k-1)}y^{-(K-k)/(k-1)}.
  \end{equation}
  These scales are at least $L$, so \cref{lem:integer_dilation_intervals} provides jointly admissible integers $s_{r,m}\in[s_{r,m}^\star,2s_{r,m}^\star]$.
  Since $w_\ast/\vartheta\geq1$, $L\geq1$, $y\leq1$, and $2\leq k\leq K$, \cref{eq:scaling_reference_size,eq:scaling_ideal_dilation} imply
  \begin{equation}\label{eq:frequency_scaling_explicit}
    \max_{(r,m)\in\mathcal I,\,1\leq j\leq K_r}|\widetilde\ell_j^{(r,m)}|
    \leq 4B^{K-2}\frac{w_\ast}{\vartheta}L^{K-1}y^{-(K-2)},
  \end{equation}
  whereas \cref{eq:scaling_weight_upper,eq:scaling_weight_lower} give
  \begin{equation}\label{eq:scaling_weight_ratio}
    \frac{w_\ast}{\vartheta} \leq e^{\mc O(K^2\log K)}R\frac{1+f_\ast}{\epsilon\|\Htarg\|}.
  \end{equation}
  Substituting \cref{eq:scaling_dilation_threshold,eq:scaling_y_bound,eq:scaling_weight_ratio} into \cref{eq:frequency_scaling_explicit}, we obtain
  \begin{equation}
    \max_{(r,m)\in\mathcal I,\,1\leq j\leq K_r}|\widetilde\ell_j^{(r,m)}|
    \leq e^{\mc O(K^4\log K)}R^{1+(K-1)^2+d(K-2)} \left(\frac{1+f_\ast}{\epsilon\|\Htarg\|}\right)^{K-1}.
  \end{equation}
  With the simple inequalities $1+(K-1)^2+d(K-2)=K^3-2K^2+2K-2\leq K^3$ and $\log K\leq K$, this proves \cref{eq:dilated_frequency_scaling}.
\end{proof}

\subsection{Explicit construction of admissible frequencies}\label{sec:spectral_spacing}

We now refine the mixed-degree construction by transferring the relative scale of the target coefficients from the physical Fourier amplitudes to the driving frequencies.
The construction proceeds in three steps.
We first discard contributions below a prescribed tolerance, then encode the remaining coefficient ratios through integer frequency dilations, and finally show that these dilations retain sufficient suppression of the higher-order Magnus terms.
This yields a Floquet construction whose non-zero physical Fourier amplitudes differ by less than a factor of two.

\begin{lemma}\label{lem:zariski_integer_points}
  Consider the zero-sum hyperplanes $H_0^{(K_r)}$ from \cref{def:H0_U}, with $d_r=(K_r-1)!$ as in \cref{def:truncated_H}, and the finite product
  \begin{equation}\label{P_set}
    \mathcal P \coloneqq \bigl(H_0^{(K_1)}\bigr)^{d_1}\times\bigl(H_0^{(K_2)}\bigr)^{d_2}\times\cdots\times\bigl(H_0^{(K_R)}\bigr)^{d_R}.
  \end{equation}
  Then the integer points of $\mathcal P$ are Zariski dense and thus every non-empty Zariski-open subset of $\mathcal P$ contains infinitely many integer points.
\end{lemma}

\begin{proof}
  For each $K$, consider the affine isomorphism
  \begin{equation}
    \varphi_K:\mathbb C^{K-1}\longrightarrow H_0^{(K)}, \qquad (z_1,\dots,z_{K-1})\longmapsto\Bigg(z_1,\dots,z_{K-1},-\sum_{j=1}^{K-1}z_j\Bigg),
  \end{equation}
  which maps $\mathbb Z^{K-1}$ bijectively onto $H_0^{(K)}\cap\mathbb Z^K$. The product map
  \begin{equation}\label{eq:bijection_phi}
    \varphi \coloneqq \bigl(\varphi_{K_1}\bigr)^{d_1} \times
                      \bigl(\varphi_{K_2}\bigr)^{d_2}\times\cdots\times\bigl(\varphi_{K_R}\bigr)^{d_R}
  \end{equation}
  is an affine isomorphism from $\mathbb C^D$ onto $\mathcal P$, mapping $\mathbb Z^D$ bijectively onto the integer points of $\mathcal P$, where
  \begin{equation}
    D\coloneqq\sum_{r=1}^{R}d_r(K_r-1).
  \end{equation}

  Using the Schwartz--Zippel Lemma~\cite{Schwartz1980}, a non-zero polynomial $f\in\mathbb C[z_1,\dots,z_D]$ of total degree $d$ cannot vanish on every point of a finite grid $S^D$ whenever $S\subset\mathbb Z$ and $|S|>d$.
  Since such grids can be chosen arbitrarily large, $\mathbb Z^D$ is Zariski dense in $\mathbb C^D$.
  By the affine isomorphism \cref{eq:bijection_phi}, it follows that the integer points are Zariski dense in $\mathcal P$, and every non-empty Zariski-open subset of $\mathcal P$ contains infinitely many of them.
  Indeed, if any of these subsets contained only finitely many integer points, removing them would leave a non-empty Zariski-open subset, contradicting Zariski density of the integer points in $\mathcal P$.
\end{proof}

To apply the above Lemma, we need to show that the subset defined by \cref{cond:admissible_frequency_family}\footnote{This is a subset of $\mathcal P$ in \cref{P_set}, since \cref{cond:admissible_frequency_family} imposes additional non-vanishing constraints on signed sums.}
is Zariski open, namely

\begin{lemma}\label{lem:separated_zariski_open}
  The frequency configurations in $\mathcal P$ satisfying \cref{cond:admissible_frequency_family} form a Zariski-open subset of $\mathcal P$.
\end{lemma}

\begin{proof}
  A configuration fails to satisfy \cref{cond:admissible_frequency_family} only if some choice of at most $K_{\max}$ labelled frequency components and signs, not corresponding to the allowed zero-sum blocks or pairwise cancellations, satisfies
  \begin{equation}
    \sum_{\nu=1}^{n}s_\nu\ell_{j_\nu}^{(r_\nu,m_\nu)}=0, \qquad n\leq K_{\max}.
  \end{equation}
  Each such equation is polynomial in the frequency components, and hence defines a Zariski-closed subset of $\mathcal P$.
  Given that only finitely many such equations arise, the configurations violating \cref{cond:admissible_frequency_family} form a Zariski-closed subset of $\mathcal P$ and their complement is thus Zariski open.
\end{proof}

To complete the argument, it remains to show that the Zariski-open set of frequency configurations satisfying \cref{cond:admissible_frequency_family} is non-empty.
We do so by constructing explicitly a sequence satisfying the separation constraints across blocks.

\paragraph{Witness construction.}
Let $K_{\max}\coloneqq\max_r K_r$ be the maximum degree, and set
\begin{equation}
  B\coloneqq K_{\max}+1.
\end{equation}
Fix an ordering of all frequency vectors $\ellvec^{(r,m)}$ appearing in \cref{cond:admissible_frequency_family}, beginning with $\ellvec^{(1,1)}$, and set $\ell_1^{(1,1)}\coloneqq1$.

Proceeding through the components of each vector $\ellvec^{(r,m)}$, once $\ell_1^{(r,m)}$ has been fixed, construct iteratively
\begin{equation}
  \ell_{j+1}^{(r,m)}\coloneqq B \ell_j^{(r,m)}, \qquad j=1,\dots,K_r-2,
\end{equation}
and choose the final frequency so that the entries sum to zero,
\begin{equation}\label{eq:balancing_term_app}
  \ell_{K_r}^{(r,m)}\coloneqq-\sum_{j=1}^{K_r-1}\ell_j^{(r,m)}.
\end{equation}
After constructing a complete vector $\ellvec^{(r,m)}$, define the first component of the next vector $\ellvec^{(r',m')}$ in the chosen ordering by
\begin{equation}
  \ell_1^{(r',m')}\coloneqq B\big|\ell_{K_r}^{(r,m)}\big|.
\end{equation}

\paragraph{Verification of the witness.}
Consider a vanishing signed multiset of at most $K_{\max}$ labelled frequency components.
First remove all pairs consisting of the same labelled component with opposite signs.
It remains to show that the reduced multiset, $\mc R$, is a disjoint union of complete blocks of the form in \cref{cond:admissible_frequency_family}.

Let $\ellvec^{(r,m)}$ be the last vector in the chosen ordering that occurs in it.
Set $k=K_r$ and $a=\ell_1^{(r,m)}$.
Let $\alpha_j\in\mathbb Z$ be the signed multiplicity of $\ell_j^{(r,m)}$ for $j=1,\dots,k-1$, let $\beta\in\mathbb Z$ be the signed multiplicity of $\ell_k^{(r,m)}$, and denote by $\Sigma_<$ the contribution of all preceding vectors.
By construction,
\begin{equation}
  \ell_j^{(r,m)}=aB^{j-1}, \qquad
  \ell_k^{(r,m)}=-a\sum_{j=0}^{k-2}B^j,
\end{equation}
and hence the vanishing relation becomes
\begin{equation}\label{eq:frequency_separation_relation}
  a\sum_{j=1}^{k-1}(\alpha_j-\beta)B^{j-1}+\Sigma_<=0.
\end{equation}

Now, note that every component belonging to a preceding vector (namely, not in $\ellvec^{(r,m)}$) has absolute value at most $a/B$.
Since $\mc R$ contains at most $K_{\max}=B-1$ components,
\begin{equation}
  |\Sigma_<|\leq\frac{B-1}{B}a<a.
\end{equation}
The first term in \cref{eq:frequency_separation_relation} is an integer multiple of $a$, and thus so must be $\Sigma_<$.
It follows that $\Sigma_<=0$, and thus
\begin{equation}\label{eq:frequency_base_L_relation}
  \sum_{j=1}^{k-1}(\alpha_j-\beta)B^{j-1}=0.
\end{equation}

Reducing \cref{eq:frequency_base_L_relation} modulo $B$ removes all terms $(\alpha_j-\beta)B^{j-1}$ except $\alpha_1-\beta$; this implies $\alpha_1-\beta\equiv0\pmod B$.
However,
\begin{equation}
  |\alpha_j-\beta|\leq|\alpha_j|+|\beta|\leq B-1,
\end{equation}
because the reduced multiset has at most $B-1$ components, hence the only possibility is $\alpha_1=\beta$.
Dividing \cref{eq:frequency_base_L_relation} by $B$ and repeating the same argument gives iteratively
\begin{equation}
  \alpha_j=\beta \qquad\text{for all }j=1,\dots,k-1.
\end{equation}
Thus all components of $\ellvec^{(r,m)}$ occur with the same signed multiplicity $\beta$, and therefore form $|\beta|$ complete copies of the zero-sum relation \cref{eq:balancing_term_app}.

Removing these (multiple) blocks for $\ellvec^{(r,m)}$ leaves the vanishing relation $\Sigma_<=0$, involving only preceding vectors.
Applying the same argument to $\Sigma_<$ and iterating shows that the reduced multiset $\mc R$ is a disjoint union of complete zero-sum blocks.
Restoring the pairwise cancellations removed at the beginning, the witness constructed above satisfies \cref{cond:admissible_frequency_family}.

\medskip

\noindent This construction is not intended to minimise the frequency magnitudes: more compressed configurations can for example be obtained with a greedy search.

\bigskip

The witness constructed above, together with \cref{lem:zariski_integer_points,lem:separated_zariski_open},  yields

\begin{corollary}\label{cor:separated_integer_frequencies}
  The integer frequency vectors satisfying \cref{cond:admissible_frequency_family} are Zariski-dense in the corresponding subset of $\mc P$, and thus there exist infinitely many such configurations.
\end{corollary}

\begin{lemma}\label{lem:integer_dilation_intervals}
  Fix $K\geq2$ and let
  \begin{equation}
    \ellvec^{(v)}\in\mathbb Z^{K_v},
    \qquad
    2\leq K_v\leq K,
    \qquad
    v=1,\dots,V,
  \end{equation}
  be reference frequency vectors up to length $K$. We recall the spacing constant in \cref{eq:integer_dilation_threshold},
  $L= 2+K^2(2K)^K V^{K-1}$;
  then, for any real numbers $s_v^\star\geq L$, there exist integers
  \begin{equation}\label{eq:integer_dilation_interval}
    s_v\in[s_v^\star,2s_v^\star]\cap\mathbb N,
    \qquad
    v=1,\dots,V,
  \end{equation}
  such that the dilated frequency vectors
  \begin{equation}
    \widetilde{\ellvec}^{(v)}
    \coloneqq
    s_v\ellvec^{(v)}
  \end{equation}
  form an admissible frequency family in the sense of \cref{cond:admissible_frequency_family}.
\end{lemma}

\begin{proof}
  We choose the multipliers $s_v$ successively.
  Suppose that $s_1,\dots,s_{v-1}$ have already been chosen so that the corresponding dilated vectors are admissible, and consider the choice of $s_v$.

  With the preceding multipliers fixed, any signed relation of length at most $K$ involving components of the new vector can be written as
  \begin{equation}\label{eq:dilation_relation}
    A s_v+B=0,
  \end{equation}
  where $A$ is a signed sum of components of $\ellvec^{(v)}$, while $B$
  depends only on the previously fixed dilated vectors.

  If $A\neq0$, \cref{eq:dilation_relation} excludes at most one value of $s_v$, and hence at most one integer; conversely for $A=0$ the signed-sum property of $\ellvec^{(v)}$ implies that the selected components of the new vector decompose into allowed zero-sum subsets.
  Removing these subsets leaves the relation $B=0$, whereby the remaining terms form an allowed relation by the induction hypothesis; if $B\neq0$, the original relation cannot vanish.
  Thus the case $A=0$ introduces no forbidden value of $s_v$.

  It remains to bound the number of values excluded by relations with $A\neq0$.
  A direct overcount is obtained by choosing an ordered list of $p\geq1$ signed components from the new vector and an ordered list of $q\geq0$ signed components from the preceding vectors, with $p+q\leq K$: there are at most $(2K)^p$ choices for the former and $(2KV)^q$ choices for the latter.
  Since $p\geq1$, we have $q\leq K-1$, and therefore the number of potentially forbidden values is at most
  \begin{equation}
    \sum_{\substack{p\geq1,\ q\geq0\\ p+q\leq K}} (2K)^p(2KV)^q \leq K^2(2K)^K V^{K-1} = L-2.
  \end{equation}

  For $s_v^\star\geq L$, the interval $[s_v^\star,2s_v^\star]$ contains at least $L$ integers.
  Since fewer than $L$ of them can be forbidden, at least one admissible integer $s_v$ exists in this interval.
  Proceeding inductively for $v=1,\dots,V$ gives the required admissible family.
\end{proof}

\subsection{Efficiency of Floquet Simulation}\label{sect:efficiency}

In the previous section, we arrived in \Cref{thm:bounded_amplitude_universal} at a characterisation of the Floquet simulation power of sets of interactions and, as an immediate consequence, a complete characterisation of Floquet-universal Hamiltonians which can simulate any possible target Hamiltonian in this way.
Moreover, this can be achieved whilst keeping all the interaction strengths $\mc O(1)$, with only the frequencies scaling as a function of the system size.
Now, an arbitrary Hamiltonian on $n$ qudits contains exponentially (in $n$) many independent parameters.
Since Floquet-universality requires the ability to simulate any Hamiltonian whatsoever, a simple counting argument implies that at least some of the parameters in any Floquet simulation construction must scale exponentially in general.
And we see this explicitly in \cref{eq:dilated_frequency_scaling}, where the frequencies depend exponentially on $K_{\max}$, the maximum degree in the Lie polynomial decomposition of the target Hamiltonian, which scales with $n$ in general.
However, for Hamiltonians that are specified by a polynomial (in $n$) number of parameters, we might hope for an efficient Floquet-simulation, where all parameters scale at most polynomially with the system size.

In this section, we show exactly this for the most important class of efficiently-specifiable Hamiltonians: $k$-local Hamiltonians on a lattice: Floquet-simulation of any such Hamiltonian on $n$ qudits can be performed with all amplitudes $\mc O(1)$ and all frequencies $\poly(n)$.
In fact, we prove this for a slightly more general class of Hamiltonians: $k$-local Hamiltonian whose non-commutativity graph has constant chromatic number, of which lattice Hamiltonians are an important special case.

To show this, we need a slight strengthening of \cref{lem:dilated_frequency_scaling}, in which frequency coefficients are reused within commuting families and constant phase shifts are introduced to keep the required integer dilations independent of the number of terms in each family.

\begin{lemma}[Frequency scaling for commuting families]\label{lem:strengthened_frequency_scaling}
  Let the target Hamiltonian be decomposed as
  \begin{equation}\label{eq:commuting_family_decomposition}
    \Htarg = \sum_{r=1}^{R}\sum_{p=1}^{n_r} \imath_{K_{r,p}}\lambda_{r,p}\LP_{r,p},
  \end{equation}
  where each $\LP_{r,p}$ is a homogeneous multilinear Lie polynomial of degree $2\leq K_{r,p}\leq K$, with real coefficients and $\lambda_{r,p}\in\mathbb R$.
  Assume that, for each fixed $r$, the sets of letters corresponding to distinct $p$ commute, i.e.
  \begin{equation}\label{eq:lattice_commute}
    [A_i^{(r,p)},A_j^{(r,q)}]=0
    \qquad
    \text{for all }p\neq q\text{ and all }i,j.
  \end{equation}
  For $2\leq k\leq K$, group together the indices~$p$ having the same polynomial degree,
  \begin{equation}
    \mathcal P_{r,k}\coloneqq\{p:K_{r,p}=k\}.
  \end{equation}

  Fix a base frequency $\omega>0$ and an error tolerance $0<\epsilon\leq1$.
  For each frequency group $g=(r,k,m,\nu)$, with $1\leq m\leq(k-1)!$ and $\nu\in\{+,-\}$, apply a common phase shift $\phi_g$, shared by all $p\in\mathcal P_{r,k}$ in that group, only to the $k$-th component of \cref{eq:H_ell_s_m_def},
  \begin{equation}
    e^{\pm i\omega\ell_k t} \longmapsto e^{\pm i(\omega\ell_k t+\phi_g)}.
  \end{equation}
  Then the reference vectors can be assigned a common set of integer dilations $s_{r,k,m}^{\nu}$, from which each group of degree~$k$ draws the $(k-1)!$ elements as required by the synthesis construction in \cref{thm:filtering}.
  The resulting bound on the dilated frequency coefficients $\widetilde\ell_j^{(r,k,m,\nu)} = s_{r,k,m}^{\nu}\ell_j^{(r,k,m)}$ is
  \begin{equation}\label{eq:strengthened_frequency_scaling}
    \max_{\substack{1\leq r\leq R,\;2\leq k\leq K,\;1\leq m\leq(k-1)!\\
        \nu\in\{+,-\},\;1\leq j\leq k}}
    |\widetilde\ell_j^{(r,k,m,\nu)}|
    \leq
    e^{CK^5}R^{K^3},
  \end{equation}
  for a constant $C>0$.

  With these frequency coefficients fixed, the physical amplitudes, phase shifts, and multi-period runtime can be chosen so that the simulation statement of \cref{eq:universal_simulation} of \cref{thm:bounded_amplitude_universal} is attained.
  In particular, \cref{eq:strengthened_frequency_scaling} is independent of the numbers $n_1,\ldots,n_R$ of Lie-polynomial terms within the commuting families.
\end{lemma}

\begin{proof}
  Use the synthesis and cutoff of \cref{thm:bounded_amplitude_per_period}, choosing the same degree-$k$ reference vectors for all $p\in\mathcal P_{r,k}$ at fixed $r$.
  The only modification is that we keep the integer dilations fixed at the level of these commuting degree-subfamilies and use one constant phase shift $\phi_g$ per frequency group to produce the factor $y^{K-k}$ needed to bring every degree-$k$ contribution to the same scale $y^K$.
  Write $w_{r,p,m}=\lambda_{r,p}\alpha_m^{(r,p)}$, let $\mathcal I$ denote the surviving triples, let $W=\max_{(r,p,m)\in\mathcal I}|w_{r,p,m}|$, and define
  \begin{equation}\label{eq:commuting_weight_split}
    u_{r,p,m}^{\pm}=1+\frac{(\pm w_{r,p,m})_+}{W},
    \qquad
    W(u_{r,p,m}^+-u_{r,p,m}^-)=w_{r,p,m},
  \end{equation}
  so that $1\leq u_{r,p,m}^{\pm}\leq2$.
  For each $g=(r,k,m,\nu)$, define the aggregated letters
  \begin{equation}
    B_j^g
    \coloneqq
    \sum_{\substack{p\in\mathcal P_{r,k}\\(r,p,m)\in\mathcal I}}
    \bigl(u_{r,p,m}^{\nu}\bigr)^{1/k}A_j^{(r,p)}.
  \end{equation}
  
  From multilinearity, expanding $\LP_m^{(k)}(B_1^g,\ldots,B_k^g)$ gives a sum over choices $p_1,\ldots,p_k$, one from each aggregated letter. If two of these indices are different, the corresponding Lie word contains letters associated with distinct values of~$p$ and therefore vanishes by the commutation relations in \cref{eq:lattice_commute}. Consequently, the only surviving contributions are those in which the same value of~$p$ is selected from every $B_j^g$. For each such contribution, the product of the prefactors is $\bigl((u_{r,p,m}^{\nu})^{1/k}\bigr)^k=u_{r,p,m}^{\nu}$, yielding
  \begin{equation}\label{eq:commuting_aggregation}
  	\LP_m^{(k)}(B_1^g,\ldots,B_k^g)
  	=
  	\sum_{\substack{p\in\mathcal P_{r,k}\\(r,p,m)\in\mathcal I}}
  	u_{r,p,m}^{\nu}\LP_{r,p,m}.
  \end{equation}
  Therefore, frequency separation is required only between at most $2R\sum_{k=2}^{K}(k-1)!$ groups, independently of the numbers $n_r$. 
  Setting the spacing constant
  \begin{equation}
    L = 2+K^2(2K)^K \left(2R\sum_{k=2}^{K}(k-1)!\right)^{K-1},
  \end{equation}
  \cref{lem:integer_dilation_intervals} gives admissible $s_{r,k,m}^{\nu}\in[L,2L]$ shared by all $p$ in the corresponding group, and using \cref{eq:scaling_reference_size}, we bound the largest frequency coefficient by
  \begin{equation}
    \max_{r,k,m,\nu,j}|\widetilde\ell_j^{(r,k,m,\nu)}|
    \leq
    4(4K)^{K-2}L
    \leq
    e^{CK^5}R^{K^3},
  \end{equation}
  proving \cref{eq:strengthened_frequency_scaling}.
  
  \medskip

  It remains to calibrate the different degrees without changing these frequencies.
  The first step is to apply to the $k$-th component of each group a phase shift $\phi_g=\pi/2+\psi_g$.
  If we set the phase offsets $\psi_g$ to zero, then the cancellation used in the proof of \cref{thm:bounded_amplitude_per_period} applies straightforwardly: each complete degree-$k$ resonance contributes reversal parity $(-1)^k$, while opposite-frequency pairs are invariant, giving the factor $(-1)^J$ instead of $(-1)^{J+1}$. Hence every unwanted Lie-polynomial contribution of degree $J\leq K$ containing such a resonance vanishes.
  
  To retain the desired leading-order contributions, we choose the physical phase offsets according to
  \begin{equation}\label{eq:commuting_group_phase}
  	\sin\psi_g = \nu \frac{(s_{r,k,m}^{\nu})^{k-1}}{2(2L)^{K-1}} y^{K-k},
  \end{equation}
  so that each resonance containing all $k$ components of a degree-$k$ group as in \cref{eq:integer_calibration_identity} occurs with the common scaling $y^K$, while \cref{eq:commuting_weight_split,eq:commuting_aggregation} reproduces the target coefficients up to the common factor $1/W$, which is absorbed into the overall scale and hence into the multi-period runtime.
  
  With the physical phase offsets of \cref{eq:commuting_group_phase}, the unwanted contributions can reappear only at degree $J\geq k+1$ and are linear in the relevant $|\psi_g|=\mc O(y^{K-k})$. The corresponding contributions are of order
  \begin{equation}
  	\mc O(y^J|\psi_g|)
  	=
  	\mc O(y^{J+K-k})
  	\subseteq
  	\mc O(y^{K+1}).
  \end{equation}
  The Magnus tail is likewise $\mc O(y^{K+1})$ by the truncation bound of \cref{thm:ACO}, thus both are one order smaller in $y$ than the target contribution, which scales as $y^K$.
  
  \medskip
 The frequency construction itself thus depends only on the number of groups, which is bounded by $2R\sum_{k=2}^{K}(k-1)!$, and is therefore independent of the numbers $n_r$. By contrast, the norms of the aggregated letters $B_j^g$, and hence the constants entering the remainder bounds, may grow with the sizes of the commuting families. In conclusion, the required value of $y$, and therefore the runtime through the multi-period construction, may depend on $n_1,\ldots,n_R$, but the frequency coefficients do not.
\end{proof}

For local Hamiltonians, the commuting-family structure required by \cref{lem:strengthened_frequency_scaling} arises naturally from a colouring of the corresponding intersection graph.
The following lemma makes this connection precise, relating the number of commuting families to its chromatic number.

\begin{lemma}\label{lem:graph_colouring_decomposition}
  Let $\Htarg=\sum_x h_x$ be any $k$-local Hamiltonian on $n$ qudits.
  Let $\cS$ be a finite set of interactions such that each $h_x\in\Lie(\cS)$.

  Then $\Htarg$ admits a decomposition of the form in \cref{eq:commuting_family_decomposition} of \cref{lem:strengthened_frequency_scaling} into homogeneous multilinear Lie polynomials over the generators $\cS$, such that, for each fixed $r$, the sets of letters corresponding to distinct terms have disjoint supports and hence satisfy the commutation condition in \cref{eq:lattice_commute}.

  Moreover, $R$ is a function only of $k$ and $\chi$, where $\chi$ is the vertex chromatic number of the intersection graph of $\{h_x\}$.
\end{lemma}

\begin{proof}
  Recall that the intersection graph of $\{h_x\}$ is the graph with vertices labelled by $x$ and an edge $(x,y)$ wherever $h_x,h_y$ act on at least one qudit in common.
  Any vertex colouring of this graph partitions the local terms $h_x$ into mutually commuting subsets corresponding to each colour $c$.
  Choose the optimal such colouring, which by definition contains $\chi$ colours.

  Label these subsets by $c=1,\dots,\chi$, and write each $h_x$ in subset $c$ as a finite sum of homogeneous multilinear Lie polynomials over $\cS$, which is possible by assumption.
  If any of the operators appearing in these Lie polynomials act on additional ancilla qudits not acted on by $h_x$ itself, introduce fresh such ancilla qudits for each $h_x$ as necessary.
  In this way, for fixed $c$, Lie-polynomial terms originating from distinct $h_x$ act on completely disjoint subsets of qudits.

  Since each $h_x$ is $k$-local, the number and maximal degree of the homogeneous multilinear terms required in this decomposition are bounded independently of $n$.
  Group the colour $c$ together with the homogeneous-term index into a multi-index $r$, relabel the corresponding local terms by $p$, and absorb their scalar coefficients and parity factors into $\lambda_{r,p}$ and $\imath_{K_{r,p}}$.
  This gives the decomposition in \cref{eq:commuting_family_decomposition}, while disjointness of the supports gives \cref{eq:lattice_commute}.
  Since the number of homogeneous terms required for each $h_x$ depends only on $k$, the resulting number $R$ of commuting families depends only on $k$ and $\chi$.
\end{proof}

Following \cref{lem:graph_colouring_decomposition}, the hypotheses of \cref{lem:strengthened_frequency_scaling} are satisfied with $R$ depending only on $k$ and $\chi$.
Together with \cref{thm:bounded_amplitude_universal}, this immediately gives:

\begin{theorem}[Floquet efficiency]\label{thm:Floquet_efficiency}
  Let $\Htarg = \sum_x h_x$ be a $k$-local Hamiltonian on $n$ qudits.
  Let $\cS$ be a finite set of interactions such that each $h_x \in \Lie(\cS)$ and $\Lie(\cS)$ is perfect.

  Then $\Htarg$ can be $(\epsilon,N)$-Floquet-simulated by an $\cS$-driven Hamiltonian of the form in \cref{def:S-driven_blocks}, whose physical amplitudes are $\mc O(1)$ and remain within a ratio of $2$, and whose dilated frequency coefficients satisfy
  \begin{equation}
    \max_{\substack{1\leq r\leq R,\;2\leq k\leq K,\;1\leq m\leq(k-1)!\\
        \nu\in\{+,-\},\;1\leq j\leq k}}
    |\widetilde\ell_j^{(r,k,m,\nu)}|
    \leq
    e^{CK^5}R^{K^3}.
  \end{equation}
  Here $R$ depends only on $\chi$ and $k$, but not on $n$, $K$ depends only on $k$, and $N=\poly(n)$.
\end{theorem}

In particular:

\begin{corollary}
  Floquet-universal Hamiltonians can efficiently (in the sense of \cref{thm:Floquet_efficiency}) Floquet-simulate all local Hamiltonians on $n$ qudits whose intersection graphs have chromatic number $\chi = \mc O(1)$ (as a function of $n$), with $\mc O(1)$ amplitudes and amplitude ratios.
\end{corollary}

Since lattice Hamiltonians are a special case of Hamiltonians whose intersection graphs have constant chromatic number, this immediately implies that:

\begin{corollary}
  Floquet-universal Hamiltonians can efficiently (in the sense of \cref{thm:Floquet_efficiency}) Floquet-simulate all local lattice Hamiltonians, with $\mc O(1)$ amplitudes and amplitude ratios.
\end{corollary}

\clearpage

\phantomsection
\addcontentsline{toc}{section}{References}
\printbibliography

\end{document}